\documentclass[11pt]{amsart}

\usepackage{amssymb}
\usepackage{amsmath}
\usepackage{amsfonts}
\usepackage{amsthm}
\usepackage{epsfig}
\usepackage{comment}
\usepackage{subcaption}
\usepackage[section]{placeins}
\usepackage{mathrsfs}
\usepackage{graphicx}
\usepackage[notrig]{physics}
\usepackage{units}

\usepackage{multirow}
\usepackage[title]{appendix}
\usepackage{xcolor}
\usepackage{textcomp}
\usepackage{manyfoot}
\usepackage{booktabs}
\usepackage{algorithm}
\usepackage{algorithmicx}
\usepackage{algpseudocode}
\usepackage{listings}
\usepackage{comment}

\newtheorem{thm}{Theorem}[section]

\theoremstyle{definition}
\newtheorem{definition}[thm]{Definition}
\newtheorem{example}[thm]{Example}

\theoremstyle{remark}

\begin{document}

\title{\textsc{Linear viscoelasticity: \\ Mechanics, analysis and approximation}}

\author{Michael Ortiz${}^{1,2}$}

\address
{
${}^{1}$Division of Engineering and Applied Science, California Institute of Technology, 1200 E.~California Blvd., Pasadena, CA 91125, USA.
\\
${}^{2}$Centre Internacional de Métodes Numerics en Enginyeria (CIMNE), Universitat Politècnica de Catalunya, Jordi Girona 1, 08034 Barcelona, Spain. 
}

\email{ortiz@caltech.edu}

\begin{abstract}
The aim of this review is to highlight the connection between well-established {\sl physical} and {\sl mathematical} principles as they pertain to the theory of linear viscoelasticity. We begin by examining the physical foundations of Boltzmann and Volterra's hereditary law formalism, and how those principles restrict the form of the hereditary law. We then turn to questions of material stability and continuous dependence on the stress history within the framework of the Lax-Milgram theorem, which we find to set forth rigorous and unequivocal conditions for the well-posedness of the linear viscoelastic problem. The outcome of this analysis is remarkable in that it gives precise meaning to fundamental physical properties such as fading memory. Finally, we turn to the question of best representation of viscoelastic materials by finite-rank hereditary operators or, equivalently, by a finite set of history or internal variables. We note that the theory of Hilbert-Schmidt operators and $N$-widths supplies the answer to the question.
\end{abstract}


\maketitle

\begin{center}
{\sl Meo dilectissimo magistro Iacobo Lubliner dedicatum, gratissimo animo.}
\end{center}

\tableofcontents

\section{Introduction}

The development of linear viscoelasticity mirrors in many ways the development of theoretical and applied mechanics over the past century and a half, and exemplifies the strong interactions between mechanics, mathematics and experimental science that undergird those developments. The aim of this review is to precisely highlight the remarkable--and mutually reinforcing--connection between well-established {\sl physical} and {\sl mathematical} principles as they pertain to linear viscoelasticity. Historical accounts of the growth of the field may be found in \cite{markovitz1977a, Golden:1989}. 

Early attempts to characterize the time-dependent behavior of materials gave rise to the science of {\sl rheology}, including viscoelasticity. Weber \cite{weber1835a, weber1841a} noted that an instantaneous extension of silk threads under load was followed by a further extension over a long period of time. Upon removal of the load, the original length was partly recovered after sufficient passage of time. From these observations, he posited the phenomena of creep, stress relaxation and damping of vibrations. The German school \cite{meyer1874a, Boltzmann:1874, boltzmann1878a, voigt1890a, Wiechert:1893, Wiechert:1899} referred to this behavior as "Elastische Nachwirkung" or "the elastic aftereffect" while the British school \cite{thomson:1865a, maxwell1867a, love1892a} spoke of the "viscosity of solids". The term "viscoelasticity" is of more recent coinage, used, e.~g., by Alfrey \cite{alfrey1948a} in the context of polymers.

The rise of polymeric materials in the post-war years rekindled interest in viscoelasticity, with seminal contributions from Eirich \cite{eirich1956a}  Staverman and Schwarzl \cite{staverman1956a} and others. The works of Eugene Bingham \cite{bingham1922}, John Ferry \cite{ferry1961}, Ronald Rivlin \cite{Rivlin:1951, rivlin1955a, rivlin1965a} and co-workers contributed to a deeper understanding of the viscoelastic properties of polymers. These studies led, among other advances, to the establishment of time-temperature superposition principles, which aim to predict long-term behavior from short-term experiments. James G. Oldroyd \cite{oldroyd1950} introduced generalized viscoelastic models incorporating non-linear effects, paving the way for more sophisticated material descriptions \cite{schapery1969a, schapery1997a, schapery2000a}. Additionally, experimental advances, including Dynamic Mechanical Analysis (DMA) \cite{Menard:2002}, nanoindentation \cite{Herbert:2008, Herbert:2009}, Dynamic Shear Resting (DST) \cite{Arbogast:1998, Bayly:2008}, Fourier Transform Rheology \cite{Wilhelm:1998}, and others, enabled researchers to characterize viscoelastic properties with higher precision.

Early work aimed at mathematical modelling of viscoelastic behavior by means of {\sl differential} constitutive relations include that of Maxwell \cite{maxwell1867a} and Meyer \cite{meyer1874a}. Maxwell introduced a model to describe the stress relaxation behavior of materials using a combination of an elastic spring and a viscous dashpot. Voigt \cite{voigt1890a} proposed a parallel spring-dashpot model, known as the Kelvin-Voigt model, to describe creep behavior. This early work laid the foundation for the subsequent development of differential models of viscoelasticity \cite{gemant1936, Fluegge:1975, tschoegl1989}. 

In sharp contradistinction to the differential models, Boltzmann \cite{Boltzmann:1874}, for isotropic materials, and later Volterra \cite{Volterra:1909}, for general anisotropic materials, proposed a general {\sl hereditary-integral} form of the constitutive relations. The approach was further developed in the thirties by the Italian school of Mathematical Physics, exemplified by Dario Graffi \cite{Fabrizio2012}. In analysis, the hereditary integral description of viscoelasticity gradually gained in acceptance over the more specialized differential form \cite{Linz:1985}. 

In the 50s, the Rational Mechanics school embarked on an effort to place continuum mechanics in general, and viscoelasticity in particular, on a rigorous physical and mathematical foundation. These efforts led to the formulation of a general theory of {\sl materials with memory}. In a series of papers starting in 1957, Green and Rivlin  \cite{green1957a, green1959a, green1960a} proposed the use of hereditary constitutive laws for the description of non-linear viscoelastic materials, originally developed by Boltzmann and Volterra for the linear case, as an alternative to models using constitutive equations of the rate type \cite{rivlin1955a}. A multiple-integral constitutive equation, arranged in a series, was also developed by Pipkin and Rogers \cite{Pipkin:1968}. A linearized version of Green and Rivlin's theory was developed by Pipkin and Rivlin \cite{pipkin1961a, pipkin1964a, rivlin1965a}. 

In practice, the entire history of a body can never be known, and the interpretation of the results of an experiment can only be justified only if the history of the system in the distant past prior to the start of the experiment has no appreciable influence on its outcome \cite{Truesdell:1965, wang1965b, perzyna1967a, lubliner1969a}. There is no unique way to render this intuitive principle of {\sl fading memory} in mathematical terms. Loosely speaking, fading memory is achieved when the memory functional is continuous in the space of histories with respect to some "obliviating topology", that poses a less restrictive characterization of "closeness" for histories whose support lies in the distant past than for histories with support lying in the recent past. 

A number of different obliviating topologies were proposed as appropriate for the characterization of the fading memory property. Coleman and Noll \cite{coleman1960a, coleman1961a, coleman1961b, coleman1962a} sought to formalize the concept of fading memory in a series of seminal papers in the sixties, see also \cite{coleman1966a, lubliner1969a}. A generalization of Coleman and Noll's theories was proposed by Wang \cite{wang1965a}, who, in his first theory, introduced the concept of an obliviating measure. In his second theory, Wang \cite{wang1965b} proposed the use the topology of compact convergence, while Perzyna \cite{perzyna1967a} proposed the use of general metric topologies. This surfeit of mathematical formalisms evinces the need for mathematical analysis to be closely directed by physical principles. 

A final epochal change has been brought on by the advent and rise of {\sl Data Science}. The availability of big material data sets, made possible by advances in experimental and computational science (see, e.~g., \cite{Sutton:2009, Buljac:2018, Schleder:2019, Bernier:2020, Wang:2024, Li:2019, Jin:2022}), has given rise to a desire to forge a closer nexus between material data and the predictions they enable. Two main paradigms have emerged, loosely corresponding to supervised and unsupervised methods in machine learning: {\sl Model-free} approaches, in which material set data is combined directly with field equations to effect predictions of quantities of interest \cite{Eggersmann:2019, Salahshoor:2023}; and {\sl model-based} approaches, in which the connection between material data and predictions is effected through the intermediate step of identifying a material law from the data \cite{Liu:2022, Bhattacharya:2023, Liu:2023, Weinberg:2023, Asad:2023, Marino:2023, Ghane:2024, Akerson:2024}. 

In the context of the second paradigm, it has been long recognized that material identification from empirical data may be regarded as an inverse problem (see, e.~g., the pioneering work of Bui \cite{Bui:1993}). For the most part, the classical work is concerned mainly with the identification of parameters in a given class of models, e.~g., polynomial expansions \cite{Pipkin:1968, pipkin1961a, pipkin1964a, rivlin1965a} or Prony series \cite{Prony:1795, Qvale:2004, Knauss:2007, Zhao:2007}, in contrast to the more challenging problem of identifying the functional form of the hereditary law itself (see, e.~g., \cite{Martins:2018} and references therein). Neural networks and machine learning have supplied a new and efficient means of representing material laws and fitting them by regression to big data sets, causing an extensive reevaluation of the field \cite{Liu:2022, Bhattacharya:2023, Liu:2023, Asad:2023, Marino:2023, Ghane:2024, Akerson:2024}.

Whereas these representations are convenient and efficient in practice, they are based on an {\sl a priori} assumption of a particular parameterized form of the hereditary law, which begs the question of what is the {\sl best}, or {\sl optimal}, representation of a given viscoelastic material, or a class of viscoelastic materials, by finite-rank hereditary operators. This problem falls squarely within the theory of {\sl n-widths} \cite{Pinkus:1985}, and was solved by Schmidt as early as 1907 \cite{Schmidt:1907}, with further seminal contributions by such giants as A.~Kolmogoroff \cite{Kolmogoroff:1936}, I.~M.~Gel'fand \cite{Gelfand:1959}, V.~M.~Tikhomirov \cite{Tikhomirov:1960}, and others. The theory extends to the case in which the hereditary law is not known exactly but it is only known to belong to a certain class of hereditary laws, e.~g., as defined by an experimental data set. 

The appeal of the theory of $N$-widths is that it supplies subspaces of histories of given dimension resulting in the best possible approximation of a class of hereditary laws. We note that the approximation of hereditary laws by finite-rank operators is in fact equivalent to the formulation of viscoelastic models in terms of a finite number of {\sl history} or {\sl internal variables} (see discussion at the end of Section~\ref{cN43Ba}). The theory of $N$-widths thus also answers the question of what is the best choice of history or internal variables for purposes of representing a given class of linear viscoelastic materials, a problem lucidly formulated in \cite{Liu:2023, Bhattacharya:2023}. 

With the history of the subject just outlined by way of guidance, in this review we begin by examining the physical foundations of Boltzmann and Volterra's hereditary law formalism,  Section~\ref{rniljj}. Specifically, we seek answers to questions such as: Why is the hereditary law in convolution form? Why is the hereditary kernel symmetric and positive definite? How does (instantaneous) viscosity differ from (delayed) rheology? (creep and relaxation). Why can we neglect the distant past in formulating the hereditary law? These properties of the hereditary law are often postulated {\sl ad hoc} and taken unquestioningly as point of departure, hence the value of a brief review.  

Section~\ref{dGmq9v} addresses questions such as: What properties of the hereditary law ensure continuity with respect to small perturbations? What properties of the hereditary law ensure material stability? (coercivity). These considerations in turn beg the question of what is a suitable functional setting in which to understand 'continuity' and 'coercivity'. This quandary is most satisfactorily solved by examining the well-posedness of the local stress-controlled problem, i.~e., the problem of determining the strain evolution of a material point given its stress evolution, and appealing to the Lax-Milgram theorem, which sets forth rigorous and unequivocal conditions for existence and uniqueness of solutions. The outcome of this analysis is remarkable in that it gives precise meaning to the ubiquitous---but seldom defined---notion of fading memory. Thus, whereas, as mentioned above, fading memory can be formally described in terms 'obliviating measures', it is often unclear what restrictions must said measures abide by and to what avail. As it turns out, the entire point of fading memory is to ensure material stability in the sense of contractivity of the forward mapping, see Section~\ref{wCAnUL}. Conversely, without fading memory a rheological material under stress could undergo runaway deformations in time and eventually disintegrate. Therefore, it is not surprising, by simple {\sl reductio ad absurdum}, that the rheological materials encountered in practice tend to exhibit fading memory. 

Section~\ref{nD4ghN} is a brief interlude concerned with verifying that, under mild conditions on the regularity of the domain and the forcing, the well-posedness of the local problem in fact carries over to the global {\sl boundary-value problem} as well. Finally, in Section~\ref{LBIbps} we turn to the question of {\sl best representation} of a given viscoelastic material, or a class of viscoelastic materials, by finite-rank hereditary operators or, equivalently, by a finite set of history or internal variables. As already mentioned, the theory of Hilbert-Schmidt operators and $N$-widths, which are reviewed in that section, supplies the answer to the question. 

Faded into the background but not forgotten are important aspects of the field of viscoelasticity such as experimental science, nonlinear analysis and computational methods. These, we leave for another day. 

\section{Linear viscoelasticity} \label{rniljj}

The linear viscoelastic solid supplies a canonical example of a linear hereditary law. The axiomatic and empirical basis of linear viscoelasticity is well established and the subject of an extensive literature (e.~g., \cite{coleman1961b, Gurtin:1962, fisher1973a}), but may stand a brief review. Specifically, we seek to ascertain the most general class of linear hereditary laws consistent with fundamental physical principles such as causality, the dissipation inequality and reciprocity, as well as with frequently encountered symmetries and properties such as time-shift invariance and integrability of the kernel. 

We consider isothermal processes throughout and omit any and all dependences on temperature for simplicity of notation. Throughout this review, we abide by the notational conventions of \cite{Rudin:1991, Evans:1998, AdamsFournier:2003}.

\subsection{Local state and evolutions} 

By {\sl local strain and stress evolutions} we understand functions, denoted $\epsilon : \mathbb{R} \to E$ and $\sigma : \mathbb{R} \to {F}$, defined over the real time line $\mathbb{R}$ with values in a finite-dimensional linear space $E$ and its dual ${F} = {E}^*$, respectively. For $t \in \mathbb{R}$, $\epsilon(t) \in E$ and $\sigma(t) \in {F}$ then denote the strain and stress at time $t$ and $\sigma(t) \cdot \epsilon(t) \in \mathbb{R}$ denotes their duality pairing. We work throughout within the linearized kinematics framework. In this setting, the terms 'strain' and 'stress' refer generally to work-conjugate variables representing the local state of deformation and the local state of internal force, respectively. However, for definiteness we confine attention to three-dimensional linear viscoelasticity and identify $E = \mathbb{R}^{3\times 3}_{\rm sym}$, the linear space of $3\times 3$ symmetric matrices. 

\subsection{The principle of superposition and hereditary laws}

An axiomatic foundation for linear viscoelasticity can be built on the Boltzmann {\sl superposition principle} \cite{Boltzmann:1874}, which asserts: 
\begin{center}
{\sl\narrower "Der Einflu\ss~der zu verschiedenen Zeiten vorhandenen 
Deformationen superponiert." 
(Every loading step makes an independent contribution to the final state.)}
\end{center}
\noindent
Thus, suppose that the stress response to a {\sl step} strain evolution
\begin{equation}
    \epsilon(t)
    =
    \left\{
        \begin{array}{ll}
            0 , & t < 0 , \\
            \epsilon_0 , & t \geq 0 ,
        \end{array}
    \right.
\end{equation}
is well-defined and equals $\bar{\mathbb{E}}(s,t)\epsilon_0$, for every $\epsilon_0 \in {E}$, $s,t \in \mathbb{R}$ and some {\sl relaxation kernel} $\bar{\mathbb{E}} : \mathbb{R} \times \mathbb{R} \to L({E},{F})$. Consider a general smooth strain evolution $\epsilon$. Then, by the principle of superposition,  approximating $\epsilon$ by piecewise constant, or {\sl simple}, functions (see, e.~g., \cite[\S1.43]{AdamsFournier:2003}, also \cite{coleman1961b, Gurtin:1962}), and passing to the limit, the corresponding stress response is given by the {\sl hereditary law}
\begin{equation} \label{64VhBF}
    \sigma(t)
    =
    \int_{-\infty}^{+\infty}
        \bar{\mathbb{E}}(s,t) \, \dot{\epsilon}(s)
    \, ds ,
\end{equation}
where $(\,\dot{}\,)$ denotes time differentiation, provided that the integral is defined, which requires sufficient regularity of the relaxation kernel and the admissible strain evolutions. 

\subsection{Causality and time-shift invariance} \label{J7hLq7}

The principle of {\sl causality} requires that the stress $\sigma(t)$ at time $t$ that arises in response to a prescribed strain evolution $\epsilon$ depend only on the {\sl past history} of strain prior to $t$. Specifically, if $\epsilon'$ and $\epsilon''$ are two prescribed strain evolutions such that $\epsilon'(s)=\epsilon''(s)$ for $s\leq t$, then, necessarily, the corresponding stresses at $t$ are equal, i.~e., $\sigma'(t)=\sigma''(t)$. Pipkin \cite[\S 6.1]{Pipkin:1991} bluntly remarks that, in particular: 
\begin{center}
\sl\narrower
"Causality means that there is no output until there is an input." 
\end{center}

\noindent
Causality requires the relaxation kernel to have the property
\begin{equation} \label{8CH6g8}
    \bar{\mathbb{E}}(s,t) = 0 ,
    \quad
    s > t ,
\end{equation}
whereupon the hereditary integral (\ref{64VhBF}) takes on the causal form
\begin{equation} \label{aQnb6M}
    \sigma(t)
    =
    \int_{-\infty}^t
        \bar{\mathbb{E}}(s,t) \, \dot{\epsilon}(s)
    \, ds .
\end{equation}
Causality is a fundamental property of physical systems and, therefore, henceforth we assume the hereditary law to be of the form (\ref{aQnb6M}) throughout.

The hereditary law (\ref{aQnb6M}) is said to be {\sl time-shift invariant} if shifting in time the strain evolution results in an equal shift in the stress response, i.~e., 
\begin{equation} 
    \sigma(t-\tau)
    =
    \int_{-\infty}^{t-\tau}
        \bar{\mathbb{E}}(s,t) \, \dot{\epsilon}(s-\tau)
    \, ds .
\end{equation}
for all admissible strain evolutions $\epsilon$ and time shifts $\tau$. Equivalently, 
\begin{equation} 
    \sigma(t)
    =
    \int_{-\infty}^t
        \bar{\mathbb{E}}(s+\tau,t+\tau) \, \dot{\epsilon}(s)
    \, ds ,
\end{equation}
which requires
\begin{equation} \label{NqmKs7}
    \bar{\mathbb{E}}(s,t) = \mathbb{E}(t-s) ,
\end{equation}
for some reduced relaxation kernel $\mathbb{E} : \mathbb{R} \to L({E},{F})$. In this representation, the causality condition (\ref{8CH6g8}) requires that
\begin{equation} \label{23ZXeY}
    \mathbb{E}(t) = 0 ,
    \quad
    t < 0 ,
\end{equation}
whereupon the hereditary law (\ref{aQnb6M}) reduces to the form
\begin{equation} \label{K3xWEG}
    \sigma(t)
    =
    \int_{-\infty}^t
        \mathbb{E}(t-s) \, \dot{\epsilon}(s)
    \, ds .
\end{equation}
Evidently, the hereditary law (\ref{aQnb6M}) is more general than (\ref{K3xWEG}) in that the former is capable of describing materials whose properties evolve in time, e.~g., due to chemical reactions, irradiation, wear, aging, and other effects \cite{lubliner1966a, lubliner1966b, lubliner1966c}.  

\subsection{Instantaneous viscosity vs. delayed response}
\label{A9evXu}

Consider the strain evolution
\begin{equation}
    \epsilon(t)
    =
    \left\{
        \begin{array}{ll}
            0, & t \leq 0 , \\
            \beta \, t , & t \geq 0 , 
        \end{array}
    \right.
\end{equation}
for some $\beta \in E$. According to the hereditary law (\ref{K3xWEG}), the corresponding stress evolution for $t \geq 0$ is
\begin{equation} 
    \sigma(t)
    =
    \Big(
        \int_0^t
            \mathbb{E}(t-s) 
        \, ds 
    \Big) \,
    \beta .
\end{equation}
The material is said to have {\sl instantaneous viscosity} if the limit
\begin{equation} \label{6mcVVN}
    M
    =
    \lim_{t\to 0^+}
    \int_0^t
        \mathbb{E}(t-s) 
    \, ds 
\end{equation}
is nonzero, in which case the stress obeys the viscosity law
\begin{equation}
    \lim_{t\to 0^+} \sigma(t) \sim M \beta ,
\end{equation}
immediately after the onset of deformation. 

We note that, for the limit (\ref{6mcVVN}) to be nonzero $\mathbb{E}$ must have an atom at the origin. In particular, $\mathbb{E}$ must be a distribution \cite{coleman1966a}. Contrariwise, if the limit (\ref{6mcVVN}) vanishes, e.~g., if $\mathbb{E}(t)$ is integrable in a neighborhood of the origin, then the material does not exhibit instantaneous viscosity, only {\sl delayed response} in the form of {\sl relaxation} and {\sl creep}. A canonical example of a material with instantaneous viscosity is the Kelvin-Voigt solid, whereas a canonical example of a material with no instantaneous viscosity is the Maxwell model \cite{maxwell1867a} or the standard-linear solid (see, e.~g., \cite{Fluegge:1975}). 

\subsection{The dissipation inequality}

The {\sl dissipation inequality} for general materials follows from the Clausius-Duhem inequality and requires that the sum of the dissipation due to the viscous stress, heat flux and internal processes be positive \cite{Coleman:1963, Truesdell:1965, Coleman:1967, lubliner1972a}. Assume that the material does not exhibit instantaneous viscosity in the sense of \S\ref{A9evXu}. Assume in addition that stresses and heat flow are uncoupled. Then, the dissipation inequality reduces to the Clausius-Planck inequality, which requires that the internal dissipation itself, i.~e., the dissipation due to the internal processes, be positive, i.~e.,
\begin{equation} \label{UYnW6J}
    - 
    \dot{\psi}(t) 
    + 
    \sigma(t) \cdot \dot{\epsilon}(t) 
    - 
    \eta(t) \cdot \dot{\theta}(t)
    \geq
    0 ,
\end{equation}
where $\psi$ is the free energy density, $\eta$ is the entropy density and $\theta$ is the absolute temperature. For the isothermal processes under consideration, the internal dissipation inequality (\ref{UYnW6J}) further reduces to
\begin{equation} \label{uXcEf6}
    - 
    \dot{\psi}(t) 
    + 
    \sigma(t) \cdot \dot{\epsilon}(t) 
    \geq
    0 .
\end{equation}
Consider now a smooth strain evolution $\epsilon(t)$ taking place within a bounded time interval $[a,b]$. Thus, the evolution defines a closed cycle of deformation with $\epsilon(a) = \epsilon(b) = 0$. Integrating (\ref{uXcEf6}) over the cycle gives
\begin{equation} \label{tGpq9j}
    \oint
        \Big(
            - 
            \dot{\psi}(t) 
            + 
            \sigma(t) \cdot \dot{\epsilon}(t) 
        \Big)
    \,dt
    \geq
    0 .
\end{equation}
But the free energy $\psi$ is a function of state, and hence it drops out from the integral (\ref{tGpq9j}), with the result
\begin{equation}
    \int_{-\infty}^{+\infty}
        \sigma(t) \cdot \dot{\epsilon}(t)
    \, dt
    \geq
    0 , 
\end{equation}
or, inserting the hereditary law (\ref{K3xWEG})
\begin{equation} 
    \int_{-\infty}^{+\infty}
        \Big( 
            \int_{-\infty}^t
                \mathbb{E}(t-s) \, \dot{\epsilon}(s)
            \, ds 
        \Big) \cdot
        \dot{\epsilon}(t) 
    \, dt 
    \geq
    0 .
\end{equation}
An application of the Fourier transform using the Parseval and convolution theorems (see, e.~g., \cite[Chapter 7]{Rudin:1991}) then gives
\begin{equation} \label{Tqr8xb}
    \frac{1}{2\pi}
    \int_{-\infty}^{+\infty}
        \omega^2 \,
        \hat{\mathbb{E}}(\omega) \,
        \hat{\epsilon}(\omega) \cdot \hat{\epsilon}^*(\omega)
    \, d\omega
    \geq
    0 ,
\end{equation}
where $(\,\hat{}\,)$ denotes Fourier transform. Demanding that (\ref{Tqr8xb}) hold for every closed deformation cycle in turn requires that $\hat{\mathbb{E}}(\omega)$ be positive definite, i.~e., 
\begin{equation} \label{DrPh8H}
    \hat{\mathbb{E}}(\omega) \, z \cdot z^* \geq 0, 
    \quad \forall z\in\mathbb{C};
    \quad \text{or} \;\;
    \hat{\mathbb{E}}(\omega) \geq 0 ,
\end{equation}
for all $\omega \in \mathbb{R}$, an inequality attributed to D.~Graffi~\cite{Graffi:1928}, see also \cite{Gurtin:1965, Fabrizio:1988, Matarazzo:2001}.

\subsection{Onsager's reciprocity principle}

Rogers and Pipkin \cite{Rogers:1963}, building on thermodynamic theories of viscoelasticity developed by Biot \cite{Biot:1954, Biot:1958}, Staverman and Schwarzl \cite{Staverman:1952, Staverman:1954} and Meixner \cite{Meixner:1953, Meixner:1954}, showed that the Onsager Reciprocity Principle \cite{Groot:1952, Coleman:1960} implies the symmetry of the relaxation kernel and suggested experimental tests to verify that property. 

To make this connection, we note that, in the Fourier representation, the hereditary law (\ref{K3xWEG}) takes the form
\begin{equation} \label{7TjTVX}
    \hat{\sigma}(\omega) = \hat{\mathbb{E}}(\omega) \hat{\beta}(\omega) ,
\end{equation}
with $\beta(t) = \dot{\epsilon}(t)$. This relation has the typical structure of a (linear) {\sl kinetic relation} between a thermodynamic flux, $\hat{\sigma}(\omega)$, and a thermodynamic force $\hat{\beta}(\omega)$, with $\hat{\mathbb{E}}(\omega)$ the phenomenological modulus relating the two \cite{Eringen:1960}. In its most abstract form, the essence of Onsager's principle is:
\begin{center}
{\sl\narrower "The kinetic relations derive from a potential."\\ 
(The second law then requires the potential to be convex.)}
\end{center}
\noindent
Evidently, the relation (\ref{7TjTVX}) derives from a kinetic potential if and only if 
\begin{equation} \label{BDEd4t}
    \hat{\mathbb{E}}^\dagger(\omega) = \hat{\mathbb{E}}(\omega) ,
\end{equation}
in which case the potential is the functional
\begin{equation}
    \Psi(\hat{\beta})
    =
    \frac{1}{2\pi}
    \int_{-\infty}^{+\infty}
        \frac{1}{2} 
        \hat{\mathbb{E}}(\omega) \,
        \hat{\beta}(\omega) \cdot \hat{\beta}^*(\omega)
    \, d\omega .
\end{equation}
As we have seen, the second law additionally requires that $\hat{\mathbb{E}}(\omega)$ be positive definite, eq.~(\ref{DrPh8H}), which renders the functional $\Psi(\hat{\beta})$ convex. We note that, by inverting the Fourier transform, the symmetry property (\ref{BDEd4t}), expressed in the frequency domain, extends to the relaxation kernel $\mathbb{E}(t)$, i.~e., it is equivalent to the symmetry property
\begin{equation} \label{JXm2jj}
    \mathbb{E}^T(t) = \mathbb{E}(t) ,
\end{equation}
in the time domain. 

\subsection{The fundamental theorem of calculus and fading memory}
Suppose that, for fixed $t$, we can write
\begin{equation}
    \bar{\mathbb{E}}(s,t) \, \dot{\epsilon}(s)
    =
    \frac{\partial}{\partial s} \Big( \bar{\mathbb{E}}(s,t) \, \epsilon(s) \Big)
    -
    \frac{\partial \bar{\mathbb{E}}}{\partial s} (s,t) \, \epsilon(s) ,
\end{equation}
with both terms integrable, and 
\begin{equation} \label{LT8uyf}
    \int_{-\infty}^t
        \frac{\partial }{\partial s} \Big( \bar{\mathbb{E}}(s,t) \, \epsilon(s) \Big)
    \, ds 
    =
    \bar{\mathbb{E}}(t,t) \, \epsilon(t) ,
\end{equation}
in the spirit of the fundamental theorem of calculus. Then, (\ref{aQnb6M}) can alternatively be written in the form
\begin{equation} \label{guT2D8}
    \sigma(t)
    =
    \bar{\mathbb{C}}(t) \, \epsilon(t)
    -
    \int_{-\infty}^t
        \bar{\mathbb{K}}(s,t)
         \, \epsilon(s)
    \, ds .
\end{equation}
where
\begin{equation} \label{TXb8Qb}
    \bar{\mathbb{C}}(t) = \bar{\mathbb{E}}(t,t) ,
    \quad
    \bar{\mathbb{K}}(s,t)
    =
    \frac{\partial\bar{\mathbb{E}}}{\partial s} (s,t) ,
\end{equation}
are the {\sl elastic moduli} and the {\sl hereditary kernel} at time $t$. In this representation, causality of the hereditary law (\ref{guT2D8}) requires
\begin{equation} \label{n3wyqZ}
    \bar{\mathbb{K}}(s,t) = 0 ,
    \quad
    s > t .
\end{equation}
We note that the form (\ref{guT2D8}) of the hereditary law is more general, or weaker, than the form (\ref{64VhBF}). Thus, given sufficient regularity of the hereditary kernel, e.~g., $\bar{\mathbb{K}}(s,t)$ essentially bounded in $s$ for a.~e.~$t$, the latter makes sense for absolutely continuous strain histories only, whereas the former makes sense for general integrable strain histories, including jumps.

We note that implied in (\ref{LT8uyf}) is an assumption that the material {\sl forgets} what happened at $t = -\infty$, which is a weak form of fading memory. Evidently, by the {\sl fundamental theorem of calculus} (\ref{LT8uyf}) is satisfied, e.~g., if there is $a > -\infty$ such that $\epsilon(t) = 0$ for $t < a$. More generally (\ref{LT8uyf}) is satisfied if the integrand is dominated by a continuous function vanishing at $-\infty$ \cite[Lemma 2.1]{Dafermos:1970}. In reference to these considerations, Fichera \cite{fichera1979b} suggestively entitled his seminal 1979 paper dedicated to the 60th birthday of Clifford A.~Truesdell: 
\begin{center}
\sl\narrower
"Avere Una Memoria Tenace Crea Gravi Problemi" \\ (having a lasting memory creates serious problems).
\end{center}
The paper, and references therein, can be consulted for lucid discussions. 

Proceeding as in \S\ref{J7hLq7}, we find that time-shift invariance of (\ref{guT2D8}) requires that
\begin{equation} \label{qgCeEo}
    \bar{\mathbb{C}}(t) = \mathbb{C} ,
    \quad
    \bar{\mathbb{K}}(s,t) = \mathbb{K}(t-s) ,
\end{equation}
for some constant elastic moduli $\mathbb{C}$ and hereditary kernel $\mathbb{K} : \mathbb{R} \to L({E},{F})$. From (\ref{TXb8Qb}), we additionally have
\begin{equation} 
    \mathbb{C} = \mathbb{E}(0) ,
    \quad
    \mathbb{K}(t)
    =
    \dot{\mathbb{E}}(t) .
\end{equation}
In this representation, the causality condition (\ref{n3wyqZ}) then requires that
\begin{equation} \label{l3S3XT}
    \mathbb{K}(t) = 0 ,
    \quad
    t < 0 ,
\end{equation}
and the reciprocity condition (\ref{JXm2jj}), the dissipation inequality (\ref{DrPh8H}) and the identities (\ref{TXb8Qb}) give
\begin{subequations} \label{R4H4Kc} 
\begin{align}
    &
    \mathbb{C}^T = \mathbb{C} ,
    \quad
    \mathbb{C} \geq 0 ,
    \\ &
    \mathbb{K}^T(\tau) = \mathbb{K}(\tau) ,
    \quad
    \mathbb{K}(\tau) \geq 0 ,
    \quad
    \tau \geq 0 .
\end{align}
\end{subequations}
Furthermore, the hereditary law (\ref{guT2D8}) reduces to the form
\begin{equation} \label{29a8Zn}
    \sigma(t)
    =
    \mathbb{C} \, \epsilon(t)
    -
    \int_{-\infty}^t
        \mathbb{K}(t-s)
         \, \epsilon(s)
    \, ds ,
\end{equation}
which may be regarded as the most general form of the hereditary law consistent with the general principles outlined in the foregoing and, together with (\ref{l3S3XT}) and (\ref{R4H4Kc}), can be taken as point of departure for all further developments.

\begin{example}[Maxwell-Wiechert model] \label{t3QNXc} {\rm In one dimension, the {\sl Maxwell-Wiechert} rheological model \cite{Wiechert:1893, Wiechert:1899} consists of one spring of stiffness $C_0$ in parallel with ${n}$ Maxwell units. Each Maxwell unit in turn consists of a spring of stiffness $C_i$ in series with a dashpot of viscosity $\eta_i$, $i=1,\dots,{n}$. The rheological model obeys the relations
\begin{subequations} \label{Rxxcp7}
\begin{align}
    & \label{5LptXv}
    \sigma(t) = C_0 \epsilon(t) + \sum_{i=1}^{n} C_i (\epsilon(t) - \epsilon_i^p(t)) ,
    \\ & \label{PMBnkp}
    \dot{\epsilon}_i^p(t) = \lambda_i (\epsilon(t) - \epsilon_i^p(t)) ,
\end{align}
\end{subequations}
with $\lambda_i = C_i/\eta_i$, where $\epsilon_i^p$, $i=1,\dots,{n}$, are the deformation evolutions of the dashpots. Eq.~(\ref{5LptXv}) represents the equilibrium response of the system, whereas (\ref{PMBnkp}) collects the rate equations for the dashpots. An elementary integration in time leads to a relation of the form (\ref{29a8Zn}), with
\begin{subequations} \label{uNx7d7}
\begin{align}
    & \label{0xrc8y}
    \mathbb{C} = C_0 + \sum_{i=1}^{n} C_i ,
    \\ & \label{Hn5a3N}
    \mathbb{K}(\tau)
    =
    \sum_{i=1}^{n} C_i \lambda_i \, {\rm e}^{-\lambda_i \tau} ,
    \quad
    \tau \geq 0 .
\end{align}
\end{subequations}
Eq.~(\ref{Hn5a3N}) is known as a {\sl Prony series} \cite{Prony:1795} and is often used to reduce and represent experimental data \cite{Qvale:2004, Knauss:2007, Zhao:2007}.  Assuming isotropic behavior, the preceding model can be extended to multiaxial deformations by recourse to a relaxation modulus of the form
\begin{equation}\label{3rIuyi}
    \mathbb{K}_{ijkl}(\tau)
    =
    \Big( {L}(\tau) - \frac{2}{3}{M}(\tau) \Big)\delta_{ij}\delta_{kl}
    +
    {M}(\tau) (\delta_{ik}\delta_{jl} + \delta_{il}\delta_{jk}) ,
    \quad
    \tau \geq 0, 
\end{equation}
with
\begin{equation}\label{yo7cHo}
    {L}(\tau)
    =
    \sum_{i=1}^l
    L_i
    {\rm e}^{-\lambda_i \tau} ,
    \quad
    {M}(\tau)
    =
    \sum_{i=1}^m
    M_i
    {\rm e}^{-\mu_i \tau} ,
    \quad
    \tau \geq 0 ,
\end{equation}
$L_i > 0$, $\lambda_i > 0$, $i=1,\dots,l$, and $M_i > 0$, $\mu_i > 0$, $i=1,\dots,m$, representing relaxation bulk and shear moduli, respectively. The corresponding elastic moduli in turn take the form
\begin{equation} \label{tvZ5uH}
    \mathbb{C}_{ijkl}
    =
    \Big( {B} - \frac{2}{3}{G} \Big) \, \delta_{ij}\delta_{kl}
    +
    {G} \, (\delta_{ik}\delta_{jl} + \delta_{il}\delta_{jk}) ,
\end{equation}
with
\begin{equation}
    {B} = B_0 + \sum_{i=1}^l \frac{L_i}{\lambda_i} ,
    \quad
    {G} = G_0 + \sum_{i=1}^m \frac{M_i}{\mu_i} ,  
\end{equation}
representing instantaneous elastic and shear moduli, respectively. 
} \hfill$\square$
\end{example}

\section{The local problem} \label{dGmq9v}

In addition to the fundamental physical properties enumerated above, we expect the hereditary law to be {\sl continuous}, in the sense that:
\begin{center}
{\sl\narrower "A 'small' change in the strain evolution should result in \\ a 'small' change in the stress evolution."}
\end{center}
\noindent
Additional restrictions on the hereditary operator follow from considerations of {\sl material stability}, which, in the present setting, is roughly the requirement that: 
\begin{center}
{\sl\narrower "A material that is unstressed at all times \\ must also be unstrained at all times."}
\end{center}
\noindent
A seemingly unrelated observation is that viscoelastic solids often have a {\sl fading memory property}, which Truesdell and Noll \cite{Truesdell:1965} enunciated as:
\begin{center}
{\sl\narrower "Events which occurred in the distant past have less influence \\ in determining the present response than those which occurred \\ in the recent past."}
\end{center}
While these general statements make eminent intuitive sense, they beg the question of how to understand 'smallness' or 'proximity', i.~e., how to topologize the spaces of strain and stress evolutions, and to what avail properties such as fading memory are required. 

We proceed to show that a precise and unambiguous characterization of these properties, and an elucidation of remarkable connections between them, can be achieved by requiring that the linear viscoelastic problem be {\sl well-posed}, i.~e., that the problem have a unique solution for a suitable class of configurations and loading conditions, and that the solution depend continuously on the data. As we shall see in Section~\ref{nD4ghN}, under reasonable assumptions on the regularity of the domain and the forcing, the question of well-posedness can be elucidated at the {\sl material point} level, i.~e., by considering the problem of determining the local strain evolution $\epsilon$ that gives rise to a given local stress evolution $\sigma$, or {\sl local stress control} problem. We therefore begin by investigating the properties of the local problem. 

\subsection{The local stress-control problem}

For given local stress evolution $\sigma : \mathbb{R} \to F$, the hereditary law (\ref{29a8Zn}) defines a {\sl convolution Volterra equation of the second kind} in the local strain evolution $\epsilon : \mathbb{R} \to E$ \cite{gripenberg:1990}. The crucial realization is that, conditions for the well-posedness of the problem, i.~e., existence, uniqueness and continuous dependence of the solutions, are conveniently supplied by the {\sl Lax-Milgram theorem} \cite[\S6.2.1]{Evans:1998}, which we proceed to examine. 

A suitable functional framework enabling this connection may be set forth as follows. We assume throughout that the properties (\ref{l3S3XT}) and (\ref{R4H4Kc}) are in force. Conveniently, we can then use the elastic moduli to metrize stresses and strains, leading to the following definitions.

\begin{definition}
[Local stress and strain spaces] We define the space $M$ of local strains as the linear space $E$ metrized by $\mathbb{C}$. We define the space $N$ of local stresses as the linear space $F$ metrized by $\mathbb{C}^{-1}$. As Euclidean spaces, $N = M^*$ and $M = N^*$ and the Riesz mapping is Hooke's law $\sigma(t) = \mathbb{C} \epsilon(t)$. 
\end{definition}

For $t \in (-\infty,T)$, the local evolution at a material point is characterized by a local strain evolution $\epsilon : (-\infty,T) \to M$ and a local stress evolution $\sigma : (-\infty,T) \to N$, where $T>0$ is a finite {\sl time horizon} beyond which the evolution of the material is of no interest. We restrict attention to local strain and stress evolutions that are square-sum\-mable over the interval of time $(-\infty,T)$. In order to allow for---and characterize---fading memory properties, we shall measure time according to a positive, continuous, non-increasing, integrable {\sl weighting function} $w : [0,+\infty) \to \mathbb{R}$, normalized to $w(0) = 1$, and denote by 
\begin{equation}
    d\mu(t) = w(T-t) \, dt ,
\end{equation}
the corresponding time measure. Weights, or {\sl influence functions}, were introduced by Mizel and Wang \cite{Mizel:1966} as a means of characterizing materials with fading memory. 

These considerations lead to the following spaces of local stress and strain evolutions. 

\begin{definition}[Spaces of local stress and strain evolutions] \label{g6PbsP} The space $\mathcal{M}$ of local strain evolutions is the weighted time-dependent Lebesgue space $L^2((-\infty,T), M, \mu)$, and the space $\mathcal{N}$ of local stress evolutions is the weighted time-dependent Lebesgue space $L^2((-\infty,T), N, \mu)$, both with the usual metrization (see, e.~g., \cite[\S5.9.2]{Evans:1998}). As Hilbert spaces, $\mathcal{N} = \mathcal{M}^*$ and $\mathcal{M} = \mathcal{N}^*$, with Riesz mapping given by the timewise application of Hooke's law $\sigma(t) = \mathbb{C} \epsilon(t)$. 
\end{definition}

In this functional framework, the hereditary law (\ref{29a8Zn}) takes the form
\begin{equation} \label{ZER2Cb}
    \sigma = \mathbb{C} (I - P) \epsilon ,
\end{equation}
with $\epsilon \in\mathcal{M}$ and $\sigma \in \mathcal{N}$, where $I \in \mathcal{B}(\mathcal{M})$ is the identity operator on $\mathcal{M}$ and 
\begin{equation} \label{6dEyzV}
    P \epsilon(t)
    :=
    \int_{-\infty}^t
        \mathbb{C}^{-1}
        \mathbb{K}(t-s)
         \, \epsilon(s)
    \, ds ,
\end{equation}
is the plastic-strain operator from $\mathcal{M}$ into itself. The local stress-control problem can then be recast in variational form by introducing the bilinear form 
\begin{equation} \label{XpR7Eg}
    a(\xi,\eta)
    :=
    ( (I-P) \xi, \eta) 
    =
    ( \xi, \eta )
    -
    ( P \xi, \eta) ,
\end{equation}
with $\xi$, $\eta \in \mathcal{M}$ and
\begin{equation}
    ( \xi, \eta )
    :=
    \int_{-\infty}^T \mathbb{C} \xi(t) \cdot \eta(t) \, d\mu(t) .
\end{equation}
Then, for $\sigma \in \mathcal{N}$ given, $\epsilon \in \mathcal{M}$ is a variational solution of problem (\ref{ZER2Cb}) if
\begin{equation} \label{DxzJy5}
    a(\epsilon,\eta)
    =
    \langle \sigma, \eta \rangle 
    :=
    \int_{-\infty}^T \sigma(t) \cdot \eta(t) \, d\mu(t),
\end{equation}
for every $\eta \in \mathcal{M}$. 

\subsection{Existence and uniqueness of solutions} \label{9YaGYg}

Conditions for the existence, uniqueness and continuous dependence of solutions of the local stress-control problem (\ref{DxzJy5}) now follow from the following application of the Lax-Milgram theorem \cite[\S6.2.1]{Evans:1998}.

\begin{thm}[Local well-posedness] \label{z4nJex}
Assume:
\begin{itemize}
\item[i)] (Elastic stability). $\mathbb{C} \in L(E,F)$, $\mathbb{C}^T = \mathbb{C}$, $\mathbb{C} > 0$.
\item[ii)] (Fading memory) There is a positive, continuous, non-increasing, integrable weighting function $w : [0,+\infty) \to \mathbb{R}$, normalized to $w(0) = 1$,  satisfying the semigroup condition
\begin{equation} \label{ttm59H}
    w(s) \geq w(s-t) w(t) ,
\end{equation}
and $\gamma < 1$ such that
\begin{equation} \label{jVy9Vz}
    \int_0^{+\infty}
        \| \mathbb{K}(\tau) \|
    \, w^{-1}(\tau) \, d\tau
    \leq
    \gamma ,
\end{equation}
where $\| \cdot \|$ denotes the operator norm.
\item[iii)] (Stress regularity) $\sigma \in \mathcal{N}$.
\end{itemize}
Then, problem (\ref{DxzJy5}) has a unique solution $\epsilon \in \mathcal{M}$ and
\begin{equation} \label{cnb4Tz}
    \|\epsilon\| \leq \frac{1}{1-\gamma} \|\sigma\| .
\end{equation}
\end{thm}

\begin{proof} \normalsize
a) \underline{Boundedness and continuity.} 
Begin by rewriting the plastic-strain operator in the standard form 
\begin{equation} \label{Tm74tX}
    P\epsilon(t) 
    =
    \int_{-\infty}^T k(t,s) \, \epsilon(s) \, d\mu(s) ,
\end{equation}
with
\begin{equation} \label{P4FxNH}
    k(t,s)
    =
    \mathbb{C}^{-1} 
    \mathbb{K}(t-s) \, H(t-s) \, w^{-1}(T-s) ,
\end{equation}
where $H(t)$ is the Heaviside step function. A sufficient condition for $P$ to be bounded (see, e.~g., \cite[Thm.~1.6]{Conway:1990}) is that there be constants $c_1>0$, $c_2>0$ such that
\begin{subequations} \label{hEg3n2} 
\begin{align}
    &
    \int_{-\infty}^T
        \| k(t,s) \| \, 
    \, d\mu(s) 
    \leq 
    c_1 ,
    \qquad
    \mu\text{--a.~e.} ,
    \\ &
    \int_{-\infty}^T
        \| k(t,s) \| \, 
    \, d\mu(t) 
    \leq
    c_2 ,
    \qquad
    \mu\text{--a.~e.} 
\end{align}
\end{subequations}
Then, 
\begin{equation} \label{R3M3sd}
    \| P \| 
    \leq
    ( c_1 c_2 )^{1/2} ,
\end{equation}
which supplies an upper bound for the operator norm of the plastic-strain operator $P$. Inserting (\ref{P4FxNH}) into (\ref{hEg3n2}) we obtain the equivalent conditions
\begin{subequations} \label{8u4G3x}
\begin{align}
    &    \label{kMC78x}
    \int_{-\infty}^t
        \| \mathbb{K}(t-s) \|
    \, ds
    \leq
    c_1 ,
    &
    \mu\text{--a.~e.} ,
    \\ &   \label{7KSg7j}
    \int_s^T
        \| \mathbb{K}(t-s) \|
        \, \frac{w(T-t)}{w(T-s)}
    \, dt
    \leq
    c_2 ,
    &
    \mu\text{--a.~e.} 
\end{align}
\end{subequations}
Estimating with the aid of (ii), we obtain 
\begin{subequations} 
\begin{align}
\begin{split}
    &
    \int_{-\infty}^t
        \| \mathbb{K}(t-s) \|
    \, ds
    =
    \int_0^{+\infty}
        \| \mathbb{K}(\tau) \|
    \, d\tau
    \leq  
    \int_0^{+\infty}
        \| \mathbb{K}(\tau) \|
        \, w^{-1}(\tau)
    \, d\tau
    \leq
    \gamma ,
\end{split}
\\
\begin{split}
    &   
    \int_s^T
        \| \mathbb{K}(t-s) \|
        \, \frac{w(T-t)}{w(T-s)}
    \, dt
    \leq 
    \int_s^T
        \| \mathbb{K}(t-s) \|
        \, w^{-1}(t-s)
    \, dt
    = \\ &
    \int_0^{T-s}
        \| \mathbb{K}(\tau) \|
        \, w^{-1}(\tau)
    \, d\tau
    \leq 
    \int_0^{+\infty}
        \| \mathbb{K}(\tau) \|
        \, w^{-1}(\tau)
    \, d\tau
    \leq
    \gamma ,
\end{split}
\end{align}
\end{subequations}
whence (\ref{8u4G3x}) is satisfied with $c_1=c_2=\gamma$. From (\ref{R3M3sd}) and Cauchy-Schwarz,
\begin{equation}
    | a(\xi,\eta) |
    \leq
    (1 + \gamma) \| \xi \| \, \| \eta \| ,
\end{equation}
which shows that the bilinear form $a(\cdot,\cdot)$ is indeed continuous.

b) \underline{Coercivity}. From (\ref{XpR7Eg}) and (\ref{R3M3sd}), 
\begin{equation} 
    a(\xi,\xi)
    =
    \| \xi \|^2
    -
    ( P \xi, \xi) 
    \geq
    (1 - \gamma) \, \| \xi \|^2 ,
\end{equation}
and the bilinear form $a(\cdot,\cdot)$ is {\sl coercive}, or $H$-{\sl elliptic}, with constant $1-\gamma > 0$. 

c) \underline{Existence, uniqueness and continuous dependence}. The claims follow from (a) and (b), and an application of the Lax-Milgram theorem (see, e.~g., \cite[\S6.2.1]{Evans:1998}).
\end{proof}

A trivial example of a weighting function satisfying the semigroup property (\ref{ttm59H}) is 
\begin{equation} \label{AJN7Uj}
    w(\tau) = {\rm e}^{- \lambda_0 \tau}, 
\end{equation}
for some limiting {\sl decay rate} $\lambda_0 > 0$. The fading memory condition (\ref{jVy9Vz}) then becomes
\begin{equation} \label{PmzATq}
    \int_0^{+\infty}
        \| \mathbb{K}(\tau) \|
    \, {\rm e}^{\lambda_0 \tau} \, d\tau
    \leq
    \gamma ,
\end{equation}
which requires $\mathbb{K}(\tau)$ to decay faster than ${\rm e}^{-\lambda_0 \tau}$ as $\tau\to+\infty$. Thus, strains occurring in the distant past have an exponentially decreasing influence on the present, which is a statement of fading memory. 

\subsection{Contractivity and material stability} \label{wCAnUL}

Theorem~\ref{z4nJex} affords a compelling unification of boundedness, continuity and fading memory concepts. Thus, the boundedness property (\ref{jVy9Vz}) of the hereditary kernel, which in turn ensures the boundedness---and hence continuity---of the plastic-strain operator, gives precise expression to the intuitive notion of fading memory, which is pervasive in the literature \cite{coleman1960a, coleman1961a, coleman1961b, coleman1962a}.  Since, in addition, $\gamma < 1$, (\ref{jVy9Vz}) sets forth a {\sl contractivity} property. This contractivity property in turn gives precise expression to the notion of {\sl material stability}, as it ensures coercivity and well-posedness of the stress-controlled problem. 

It is readily verified that the contractivity condition indeed renders the Picard iteration mapping \cite{Linz:1985}
\begin{equation} \label{7VNCm8}
    \epsilon 
    \mapsto
    \mathbb{C}^{-1} \sigma + P \epsilon
\end{equation}
a contraction from $\mathcal{M}$ to itself. This iteration is the basis of the classical method of {\sl equivalent body forces} \cite{Mendelson:1968}. By the Banach fixed point theorem \cite[\S9.2.1]{Evans:1998}, the mapping then has a fixed point satisfying (\ref{ZER2Cb}), which comports with the conclusions of the Lax-Milgram theorem. Indeed, the Lax-Milgram theorem can be proved using the Banach fixed point theorem \cite[\S6.2.1]{Evans:1998}.

\begin{example}[Maxwell-Wiechert model] {\rm 
We investigate the contractivity properties of the Maxwell-Wiechert model introduced in Example~\ref{t3QNXc}. From (\ref{tvZ5uH}), for every $\xi \in M$ we have
\begin{equation} \label{bbEwiB}
    \| \xi \|^2
    =
    \mathbb{C} \, \xi \cdot \xi
    = 
    {B} \, {\rm tr}(\xi)^2 
    +
    2 {G} \xi'\cdot\xi' ,
    \quad
    \xi' := \xi - \frac{1}{3} {\rm tr}(\xi) \, I ,
\end{equation}
and, for every $\eta \in N$, 
\begin{equation} \label{L2H3Io}
    \| \eta \|^2
    =
    \mathbb{C}^{-1} \, \eta \cdot \eta
    = 
    \frac{1}{9 {B}} \, {\rm tr}(\eta)^2 
    +
    \frac{1}{2 {G}} \, \eta'\cdot\eta' ,
    \quad
    \eta' := \eta - \frac{1}{3} {\rm tr}(\eta) \, I .
\end{equation}
Testing $\mathbb{K}(\tau)$ with $\xi \in M$ such that $\xi' = 0$, using (\ref{bbEwiB}) and (\ref{L2H3Io}), we find
\begin{equation}
    \frac{\|\mathbb{K}(\tau) \xi\|}{\|\xi\|}
    =
    \frac{L(\tau)}{B} .
\end{equation}
Likewise, testing $\mathbb{K}(\tau)$ with $\xi \in M$ such that ${\rm tr} \xi = 0$, using (\ref{bbEwiB}) and (\ref{L2H3Io}), we find
\begin{equation}
    \frac{\|\mathbb{K}(\tau) \xi\|}{\|\xi\|}
    =
    \frac{M(\tau)}{G} .
\end{equation}
From these identities and volumetric/deviatoric orthogonality, it follows that
\begin{equation} \label{ThnZsq}
    \|\mathbb{K}(\tau)\|
    =
    \max \Big\{ \frac{L(\tau)}{B}, \frac{M(\tau)}{G} \Big\} .
\end{equation}
The same result is obtained evaluating $\|\mathbb{K}(\tau)\|$ directly as the square root of the maximum eigenvalue of $\mathbb{K}^T(\tau) \mathbb{C}^{-1} \mathbb{K}(\tau)$ with respect to $\mathbb{C}$. Suppose, for simplicity, that the material is incompressible, corresponding to the limit $B_0 \uparrow +\infty$. Then, (\ref{ThnZsq}) simplifies to
\begin{equation} \label{bxE3NW}
    \|\mathbb{K}(\tau)\|
    =
    \frac{M(\tau)}{G}
\end{equation}
Choose, for definiteness, an exponential weight of the form (\ref{AJN7Uj}), with
\begin{equation}
    \lambda_0 
    <
    \min_{1\leq i \leq m} \mu_i .
\end{equation}
Then, from (\ref{bxE3NW}) we have
\begin{equation} \label{TNEmVA}
    \int_0^{+\infty}
        \|\mathbb{K}(\tau)\| \, w^{-1}(\tau)
    \, d\tau
    =
    \sum_{i=1}^m \frac{M_i/(\mu_i-\lambda_0)}{G} 
\end{equation}
Set $\lambda_0=0$. Then, (\ref{TNEmVA}) reduces to 
\begin{equation}
    \int_0^{+\infty}
        \|\mathbb{K}(\tau)\| \, w^{-1}(\tau)
    \, d\tau
    =
    \frac{\sum_{i=1}^m M_i/\mu_i}{G_0+\sum_{i=1}^m M_i/\mu_i} 
    <
    1 ,
\end{equation}
since $G_0>0$. By continuity, it follows that there is an interval of values of $\lambda_0$ where the contractivity condition $\gamma < 1$ is satisfied. 
} \hfill$\square$
\end{example}

\section{The viscoelastic boundary-value problem} \label{nD4ghN}

So far, we have restricted attention to local material behavior, i.~e., the response of a single material point. However, it is straightforward to see that the properties of the local material behavior in fact carry over, {\sl mutatis mutandi}, to the global level, i.~e., to boundary value problems governing the response of solids and structures (see, e.~g., \cite{Golden:1989} for background).

To this end, consider a viscoelastic solid occupying a domain $\Omega \subset \mathbb{R}^{n}$ with state defined by a time-dependent displacement field $u : \Omega \times (-\infty,T) \to \mathbb{R}^{n}$. The governing compatibility and equilibrium laws are
\begin{subequations}\label{sIes1A}
\begin{align}
    &   \label{mPn7kf}
    \epsilon(x,t)
    =
    \frac{1}{2} \big( D u(x,t) + D u^T(x,t) \big) ,
    &  \text{in } \Omega \times (-\infty,T) ,
    \\ & \label{eqGammaD}
    u(x,t) = g(x,t) ,
    &   \text{on } \Gamma_D \times (-\infty,T) ,
\end{align}
\end{subequations}
and
\begin{subequations}\label{fRoa1l}
\begin{align}
    &   \label{DuVp7m}
    \operatorname{\rm div} \sigma(x,t) + f(x,t) = 0 ,
    &
    \text{in } \Omega \times (-\infty,T),
    \\ &
    \sigma(x,t) \, \nu(x) = h(x,t)  ,
    &  \label{eqGammaN}
    \text{on } \Gamma_N \times (-\infty,T),
\end{align}
\end{subequations}
respectively, where $t \in (-\infty,T)$, $\epsilon : \Omega \times (-\infty,T) \to E$ is a time-dependent strain field, $\sigma : \Omega \times (-\infty,T) \to F$ is a time-dependent stress field, $f : \Omega \times (-\infty,T) \to \mathbb{R}^{n}$ are body forces, $g : \Gamma_D \times (-\infty,T) \to \mathbb{R}^{n}$ are displacements prescribed on the Dirichlet boundary $\Gamma_D$, $h : \Gamma_N \times (-\infty,T) \to \mathbb{R}^{n}$ are tractions applied on the Neumann boundary $\Gamma_N$, and $\nu : \partial\Omega \to S^{n-1}$ denotes the unit outward normal. In addition, we assume that the hereditary law (\ref{29a8Zn}) is satisfied pointwise at material points, i.~e.,
\begin{equation} \label{X7C2hr}
    \sigma(x,t)
    =
    \mathbb{C} \epsilon(x,t) 
    - 
    \int_{-\infty}^t
        \mathbb{K}(t-s) \epsilon(x,s)
    \, ds ,
    \quad
    \text{a.~e.~in} \;\, \Omega \times (-\infty,T) .
\end{equation}
We wish to ascertain conditions under which a displacement evolution $u(x,t)$ exists and is unique for given forcing $f(x,t)$, $g(x,t)$ and $h(x,t)$. 

One first concern is to identify restrictions on the domain $\Omega$, the Dirichlet boundary $\Gamma_D$ and the Neumann boundary $\Gamma_N$ ensuring that the operations expressed in (\ref{sIes1A}), (\ref{fRoa1l}) and (\ref{X7C2hr}) are well-defined in the sense of traces. A detailed treatment can be found in \cite[\S1.2 and 1.3]{Temam:1979}. For present purposes, it suffices to assume that $\Omega$ is a connected, open, bounded, nonempty Lipschitz set, and that $\Gamma_D$, $\Gamma_N$ are disjoint open subsets of $\partial\Omega$ such that
\begin{equation}\label{TX5yXY}
    \overline{\Gamma}_D \cap \overline{\Gamma}_N = \partial\Omega,\hskip3mm
    \mathcal{H}^{n-1}(\overline\Gamma_N\setminus\Gamma_N)
    =
    \mathcal{H}^{n-1}(\overline\Gamma_D\setminus\Gamma_D)=0,
    \hskip3mm\Gamma_D\ne\emptyset ,
\end{equation}
where $\mathcal{H}^{n-1}$ is the Hausdorff measure of dimension $n-1$. If, in addition, $\Gamma_D$ and $\Gamma_N$ are Lipschitz subsets of $\partial\Omega$, then by translation we may take $g=0$ without loss of generality (see, e.~g., \cite[\S2.1]{Conti:2018}).

These considerations lead to the following definition.

\begin{definition}[Space of global evolutions] \label{qk2mTK} For $t \in (-\infty,T)$ and $\Omega$, $\Gamma_D$ and $\Gamma_N$ sufficiently regular, the space of global displacement fields is
\begin{equation}
    V = 
    \{
        u \in H^1(\Omega; \mathbb{R}^{n}) \, : \,
        u = 0 \text{ on } \Gamma_D
    \} .
\end{equation}
Furthermore, the space of global displacement evolutions is the Bochner space $\mathcal{V} = L^2((-\infty,T), V, \mu)$ with the usual metrization (see, e.~g., \cite[\S5.9.2]{Evans:1998}). When convenient, we shall use the notation $u(x,t) := u(t)(x)$ for elements $u \in \mathcal{V}$.
\end{definition}

As in the local case, we can recast the problem (\ref{sIes1A}), (\ref{fRoa1l}), (\ref{X7C2hr}) in variational form as
\begin{equation} \label{LZ9JWu}
    a(u,v) - b(v) = 0, \quad \forall v \in \mathcal{V} , 
\end{equation}
where now we write
\begin{equation} \label{3brgWV}
\begin{split}
    & 
    a(u,v) 
    := 
    \int_{-\infty}^T
        \Big\{
            \int_\Omega
                \Big(
                    \mathbb{C} E u(x,t) 
                    - 
                    \int_{-\infty}^t
                        \mathbb{K}(t-s) E u(x,s)
                    \, ds
                \Big)
                E v(x,t)
            \, dx
        \Big\}
    \, dt ,
\end{split}
\end{equation}
with 
\begin{equation}\label{Hlusp6}
    E u = \frac{1}{2}(D u + D u^T) ,
\end{equation}
for the strain operator, and 
\begin{equation} \label{6dX954}
\begin{split}
    & 
    b(v)
    := 
    \int_{-\infty}^T
        \int_\Omega
            f(x,t) \cdot v(x,t)
        \, dx
    \, dt
    + 
    \int_{-\infty}^T
        \int_{\Gamma_N}
            h(x,t) \cdot v(x,t)
        \, d\mathcal{H}^{n-1}(x)
    \, dt ,
\end{split}
\end{equation}
with the time integrals understood in the sense of Bochner \cite[\S E.5]{Evans:1998}

Again, conditions for existence and uniqueness of solutions, as well as continuous dependence on the data, are set forth by the Lax-Milgram theorem \cite[\S6.2.1]{Evans:1998}.

\begin{thm} [Existence, uniqueness, continuous dependence] \label{f060yF}
Let $\Omega \subset \mathbb{R}^{{n}}$ be nonempty, connected, open, bounded and Lipschitz. Let $\Gamma_D$, $\Gamma_N$ be disjoint open Lipschitz subsets of $\partial\Omega$ satisfying (\ref{TX5yXY}). Suppose that assumptions (i) and (ii) of Theorem~\ref{z4nJex} hold. Suppose, in addition, that $f \in L^2((-\infty,T), H^{-1}(\Omega; \mathbb{R}^{n}),\mu)$ and $h \in L^2((-\infty,T), H^{-1/2}(\Gamma_N; \mathbb{R}^{n}), \mu)$. Then, the problem (\ref{LZ9JWu}) has a unique solution $u \in \mathcal{V}$ and
\begin{equation} \label{3MyDYE}
    \| u \| 
    \leq 
    C \,
    \big(
        \| f \|_{L^2((-\infty,T), H^{-1}(\Omega;\mathbb{R}^n), \mu)}
        +
        \| h \|_{L^2((-\infty,T), H^{-1/2}(\partial\Omega;\mathbb{R}^n), \mu)}
    \big) ,
\end{equation} 
for some $C>0$. 
\end{thm}

\begin{proof} \normalsize 
a) \underline{Continuity.} Let $u, v \in \mathcal{V}$. Then, $Eu, Ev \in {L^2((-\infty,T), L^2(\Omega,E), \mu)}$. By Bochner's theorem \cite[Theorem E.5.7]{Evans:1998} and (\ref{R3M3sd}), we have 
\begin{equation}
\begin{split}
    | a(u,v) |
    & \leq 
    \| E u \| \, 
    \| E v \|
    + 
    \| P E u \| \, 
    \| E v \|
    \leq 
    (1 + \gamma) \,
    \| E u \| \, 
    \| E v \| .
\end{split}
\end{equation}
By Poincar\'e's inequality (see, e.~g., \cite[Theorem~6.30]{AdamsFournier:2003}), $\| E u \|$ is dominated by $\| u \|$ in $\mathcal{V}$ and
\begin{equation}
\begin{split}
    | a(u,v) |
    & \leq 
    (1 + \gamma) \, C \,
    \| u \| \, \| v \| ,
\end{split}
\end{equation}
for some $C > 0$. In addition, by Bochner's theorem \cite[Theorem E.5.7]{Evans:1998}, the trace theorems and the assumed regularity of $f$ and $h$ it follows that $b \in \mathcal{V}^*$. 

b) \underline{Coercivity.} Proceeding as in Theorem~\ref{z4nJex} and appealing to Bochner's theorem \cite[Theorem E.5.7]{Evans:1998}, we obtain
\begin{equation} \label{3T2wfd}
    a(u,u) 
    \geq 
    (1-\gamma) \,
    \| E u \|^2 .
\end{equation}
Korn's inequality \cite{Ciarlet:2010} then ensures that $\| E u \|$ dominates $\| u \|$ in $\mathcal{V}$, and 
\begin{equation}
    a(u,u) 
    \geq 
    (1-\gamma) \, C \, \| u \|^2 ,
\end{equation}
for some $C > 0$, which proves coercivity.

c) \underline{Continuous dependence.} Existence and uniqueness are a consequence of the Lax-Milgram theorem \cite[\S6.2.1]{Evans:1998}, (a) and (b). In addition, testing (\ref{LZ9JWu}) with $v=u$ and using (\ref{3T2wfd}), we obtain
\begin{equation}
\begin{split}
    (1-\gamma) \, C \, \| u \|^2 
    \leq
    a(u,u) 
    =
    b(u)
    \leq \| b \| \, \|u\| .
\end{split}
\end{equation}
Estimating (\ref{6dX954}), 
\begin{equation}
    \| b \|
    \leq
    C \,
    \big(
        \| f \|_{L^2((-\infty,T), H^{-1}(\Omega;\mathbb{R}^n), \mu)}
        +
        \| h \|_{L^2((-\infty,T), H^{-1/2}(\partial\Omega;\mathbb{R}^n), \mu)} 
    \big) ,
\end{equation}
for some $C>0$, whence (\ref{3MyDYE}) follows, as advertised. 
\end{proof}

The main significance of the preceding theorem is twofold. The theorem shows that, under mild regularity assumptions on the domain and the forcing, the local properties of the hereditary law, as stipulated in Theorem~\ref{z4nJex}, ensure well-posedness at the global level. In addition, the theorem shows that the functional framework under consideration
indeed sets forth a natural topology for the evolutions of linear viscoelastic solids. Finally, the theorem gives us license to focus attention at the local level, which we do henceforth. 

\section{Approximation of the hereditary operator} \label{LBIbps}

As already noted, a far-reaching consequence of the preceding analysis is the identification of a suitable functional framework in which to elucidate questions of convergence and approximation. One type of approximation concerns numerical schemes aimed at generating converging sequences of solutions \cite{Linz:1985, Tallec:1990, Reese:1998, Govindjee:2004}. By contrast, here we focus on approximation of the hereditary kernel itself, in the spirit of experimental data reduction and constitutive identification. 

\subsection{Approximation of the plastic-strain operator}
Suppose that an approximation scheme results in a sequence $(P_h)$ of approximate plastic-strain operators converging to $P$ in the operator norm. Suppose that the operators $(P_h)$ satisfy the conditions of Theorem~\ref{z4nJex} uniformly with constant $\gamma$. Let $\epsilon_h \in \mathcal{M}$ be the solution of the approximate problem
\begin{equation} \label{fZM2ZA}
    \sigma = \mathbb{C} (I - P_h) \epsilon_h .
\end{equation}
Then, it is readily verified that the sequence $(\epsilon_h)$ converges to $\epsilon$ in $\mathcal{M}$. Thus, subtracting (\ref{fZM2ZA}) from (\ref{ZER2Cb}) gives
\begin{equation}
    \epsilon_h - \epsilon
    =
    P_h \epsilon_h - P\epsilon ,
\end{equation}
and, estimating, 
\begin{equation}
\begin{split}
    \|\epsilon_h-\epsilon\|
    & =
    \|P_h \epsilon_h- P\epsilon \|
    \\ & = 
    \|P_h (\epsilon_h-\epsilon) + (P_h-P)\epsilon \|
    \leq 
    \gamma \, 
    \|\epsilon_h-\epsilon\| 
    + 
    \|P_h-P\| \, \| \epsilon \| .
\end{split}
\end{equation}
Rearranging terms,
\begin{equation} \label{8dcSfH}
    \|\epsilon_h-\epsilon\|
    \leq
    \frac{1}{1-\gamma} \, 
    \|P_h-P\| \, \| \epsilon \| ,
\end{equation}
which shows $\epsilon_h \to \epsilon$ in $\mathcal{M}$ if $P_h \to P$ in the operator norm, as advertised. 

It therefore follows that the local stress-controlled problem is stable with respect to perturbations in the plastic-strain operator provided that the perturbations are controlled in the operator norm, which sets forth the appropriate notion of convergence of approximations. 

\subsection{Approximation of the memory kernel}

A particular case of interest concerns the approximation of the hereditary kernel $\mathbb{K}(\tau)$. Thus, suppose that the approximate plastic-strain operators are of the form
\begin{equation} \label{AwKzd6}
    P_h \epsilon(t)
    :=
    \int_{-\infty}^t
        \mathbb{C}^{-1}
        \mathbb{K}_h(t-s)
         \, \epsilon(s)
    \, ds ,
\end{equation}
for some sequence $(\mathbb{K}_h)$ of hereditary kernels. In this case, we expect that the convergence of solutions be controlled directly by the convergence of the sequence $(\mathbb{K}_h)$ in some suitable sense, to be determined. Indeed, subtracting (\ref{6dEyzV}) from (\ref{AwKzd6}), we obtain
\begin{equation} 
\begin{split}
    (P_h-P) \epsilon(t)
    & =
    \int_{-\infty}^t
        \mathbb{C}^{-1}
        \big(\mathbb{K}_h(t-s)-\mathbb{K}(t-s)\big)
         \, \epsilon(s)
    \, ds .
\end{split}
\end{equation}
Proceeding as in (\ref{8u4G3x}), we estimate
\begin{subequations} 
\begin{align}
\begin{split}
    &   
    \int_s^T
        \| \mathbb{K}_h(t-s)-\mathbb{K}(t-s) \|
        \, \frac{{w(T-t)}}{{w(T-s)}}
    \, dt
    \leq 
    \int_0^{+\infty}
        \| \mathbb{K}_h(\tau) - \mathbb{K}(\tau) \|
        \, w^{-1}(\tau)
    \, d\tau ,
\end{split}
    \\ 
\begin{split}
    &
    \int_{-\infty}^t
        \| \mathbb{K}_h(t-s)-\mathbb{K}(t-s) \|
        \, \frac{{w(T-t)}}{{w(T-s)}}
    \, ds
    \leq 
    \int_0^{+\infty}
        \| \mathbb{K}_h(\tau) - \mathbb{K}(\tau) \|
        \, w^{-1}(\tau)
    \, d\tau ,
\end{split}
\end{align}
\end{subequations}
and, by (\ref{R3M3sd}),
\begin{equation} \label{8svMvb}
    \| P_h - P \| 
    \leq
        \int_0^{+\infty}
        \| \mathbb{K}_h(\tau) - \mathbb{K}(\tau) \|
        \, w^{-1}(\tau)
    \, d\tau .
\end{equation}
Therefore, convergence of solutions is ensured if the sequence $(\mathbb{K}_h)$ of memory kernels converges to $\mathbb{K}$ in the space $\mathcal{K} := L^1((0,+\infty), L(M), w^{-1}(\tau) \, d\tau)$, which thus arises as the natural space of memory kernels and establishes a powerful link with the problem of approximation of $L^1$ functions \cite{Pinkus:1989}.

\subsection{Approximation of the relaxation spectrum} \label{ZRp3LD}

The notion that the rheology of materials arises from the superposition of internal mechanisms, each characterized by a {\sl relaxation time}, was introduced by Wiechert \cite{Wiechert:1893}, and pervades much of the theory and praxis of linear viscoelasticity. The collection of relaxation times, or {\sl relaxation spectrum}, can be finite, countable or continuous (see, \cite[\S4]{tschoegl1989}; also Kestin and Rice \cite{Kestin:1970} for a critical review). 

To elucidate this connection, begin by writing a general hereditary kernel in the form
\begin{equation} \label{eqZ2Sh}
    \mathbb{K}(\tau)
    =
    \int_{\lambda_0}^{+\infty}
        (\lambda-\lambda_0) \, 
        {\rm e}^{-\lambda \tau}
    \, d\nu(\lambda) ,
    \quad
    \tau \geq 0 ,
\end{equation}
where $1/\lambda$ is a generic relaxation time and $\nu$ is an $L(M,N)$-valued measure with support in $[\lambda_0,+\infty)$, or {\sl relaxation measure}, with $1/\lambda_0$ a cutoff relaxation time, possibly infinite. The support of $\nu$ is the {\sl relaxation spectrum}. Alternative spectral representations can be based on the Laplace transform \cite[\S4]{tschoegl1989}. Estimating and appealing to Fubini's theorem, we find
\begin{equation}
\begin{split}
    &
    \int_0^{+\infty}
        w^{-1}(\tau) \, \|\mathbb{K}(\tau)\|
    \, d\tau
    = \\ &
    \int_0^{+\infty}
        w^{-1}(\tau) \,
        \|
            \int_{\lambda_0}^{+\infty}
                (\lambda-\lambda_0) \, 
                {\rm e}^{-\lambda \tau}
            \, d\nu(\lambda) 
        \|
    \, d\tau
    \leq \\ &
    \int_0^{+\infty}
        w^{-1}(\tau) \,
        \Big(
            \int_{\lambda_0}^{+\infty}
                (\lambda-\lambda_0) \, 
                {\rm e}^{-\lambda \tau}
            \, d\|\nu\|(\lambda) 
        \Big)
    \, d\tau
    \leq \\ &
    \int_{\lambda_0}^{+\infty}
        (\lambda-\lambda_0) \, 
        \Big(
            \int_0^{+\infty}
                w^{-1}(\tau) \,
                {\rm e}^{-\lambda \tau}
            \, d\tau
        \Big)
    \, d\|\nu\|(\lambda) .
\end{split}
\end{equation}
From this estimate we conclude that
and $\mathbb{K}$ belongs to $\mathcal{K}$ provided that 
\begin{equation} \label{P4dtn6}
    (\lambda-\lambda_0) 
    \int_{\lambda_0}^{+\infty}
        w^{-1}(\tau) \,
        {\rm e}^{-\lambda \tau}
    \, d\tau
    \leq
    C < +\infty ,
\end{equation}
and the relaxation measure $\nu$ is bounded. For instance, if $w(\tau) = 1$, then we can take $\lambda_0 = 0$, corresponding to an infinite cutoff relaxation time, whence the uniform bound (\ref{P4dtn6}) is satisfied with $C=1$.

Suppose now that the relaxation spectrum is approximated by a sequence $(\nu_h)$ of relaxation measures, in turn generating a sequence of hereditary kernels
\begin{equation} \label{gTUf4H}
    \mathbb{K}_h(\tau)
    =
    \int_{\lambda_0}^{+\infty}
        (\lambda-\lambda_0) \, 
        {\rm e}^{-\lambda \tau}
    \, d\nu_h(\lambda) ,
    \quad
    \tau \geq 0 .
\end{equation}
Suppose that the sequence $(\nu_h)$ converges weakly to $\nu$ in the sense of measures. Then, 
\begin{equation}
    \mathbb{K}_h(\tau) - \mathbb{K}(\tau)
    =
    \int_{\lambda_0}^{+\infty}
        (\lambda-\lambda_0) \,
        {\rm e}^{-\lambda \tau}
    \, d(\nu_h-\nu)(\lambda) 
    \to
    0 .
\end{equation}
Hence $\mathbb{K}_h(\tau)$ converges to $\mathbb{K}(\tau)$ pointwise, and
\begin{equation} \label{XDtE34}
    \int_0^{+\infty}
        \| \mathbb{K}_h(\tau) - \mathbb{K}(\tau) \|
        \, w^{-1}(\tau)
    \, d\tau 
    \to
    0 ,
\end{equation}
by dominated convergence \cite[Theorem~1.50]{AdamsFournier:2003}. 

We note that the preceding analysis based on a relaxation measure extends the classical treatment based on continuous densities \cite[\S4]{tschoegl1989} and unifies the treatment of discrete and continuous spectra. It also elucidates the appropriate notion of approximation and convergence of spectral representations, namely, weak convergence of the relaxation measure.  

\begin{example}[Approximation property of Prony series] \label{ze9Gs3} {\rm
{\sl Prony series}, see Example~\ref{t3QNXc}, are widely used to represent and approximate memory kernels, and to reduce experimental data \cite{Qvale:2004, Knauss:2007, Zhao:2007}. The approximation properties of Prony series can be established as follows. The Prony series model (\ref{yo7cHo}) corresponds to approximate relaxation measures of the form
\begin{equation} \label{m9nwCt}
    \nu_h 
    = 
    \sum_{i=1}^{N_h} \mathbb{C}_i \delta_{\lambda_i} ,
\end{equation}
where $\mathbb{C}_i \in L(E,F)$ and $\delta_{\lambda_i}$ is the Dirac measure centered at $\lambda_i > \lambda_0$. Discrete measures of the form (\ref{m9nwCt}) are weakly dense in the space of measures \cite[\S8.1.5]{Bogachev:2007}. Thus, it is always possible to find an approximating sequence $(\nu_h)$ converging weakly to a given relaxation measure $\nu$. By the preceding analysis, the corresponding sequence of Prony series $\mathbb{K}_h$, eq.~(\ref{gTUf4H}), converges to the hereditary kernel $\mathbb{K}$ generated by $\nu$, eq.~(\ref{eqZ2Sh}), in $\mathcal{K}$, which establishes the approximating property of Prony series. An alternative--and equivalent--interpretation of this approximation property is to observe that Dirac measures are extreme points in the space of measures and appeal to the Krein-Milman theorem \cite[Thm.~3.23]{Rudin:1991}. 
}\hfill$\square$
\end{example}

\section{History representation}

In view of the {\sl Hilbert-space structure} of the spaces $\mathcal{M}$ and $\mathcal{N}$, it is natural to represent strain and stress evolutions with reference to a {\sl basis}. Evidently, the choice of basis depends critically on the weighting function $w(t)$, which in turn depends on the decay properties of the hereditary kernel $\mathbb{K}(\tau)$, see Theorem~\ref{z4nJex}. Remarkably, the elucidation of this question, and other seemingly disconnected representations, such as internal variables (see \cite{Horstemeyer:2010} for a historical overview), leads directly to the theory of {\sl Hilbert-Schmidt operators} \cite{Conway:1990, Rudin:1991, Brezis:2010} and {\sl $N$-widths} in Hilbert spaces \cite{Pinkus:1985}, as we shall see next.

\subsection{Basis representation} \label{cN43Ba}

In order to standardize the representation of evolutions, it proves convenient to adopt a {\sl history representation}. For a fixed material point, the {\sl past local histories} of strain and stress up to time $t$ are the functions 
\begin{equation}
    \epsilon_t(\tau) = \epsilon(t-\tau) ,
    \quad
    \sigma_t(\tau) = \sigma(t-\tau) , 
    \quad
    \tau \geq 0 . 
\end{equation}
In terms of histories, for fixed $t$ the hereditary law (\ref{ZER2Cb}) becomes
\begin{equation} \label{bgmUJR}
    \sigma_t(\tau)
    =
    \mathbb{C} \, ( I - S ) \, \epsilon_t(\tau) ,
\end{equation}
where the plastic-strain history operator 
\begin{equation} \label{yCn5JN}
    S \, \epsilon_t(\tau)
    =
    \int_\tau^{+\infty} 
        \mathbb{C}^{-1} \mathbb{K}(\rho-\tau) \, 
        \epsilon_t(\rho) 
    \, d\rho , 
\end{equation}
maps local histories of strain $\epsilon_t$ in the space of local strain histories $\mathcal{H}$ $:=$ $L^2((0,+\infty), M, \mu)$ to local histories of plastic strain $S \, \epsilon_t$ also in $\mathcal{H}$.

\begin{definition}[Spaces of local stress and strain histories] \label{D3jOho
} The space $\mathcal{H}$ of local strain histories is the weighted time-dependent Lebesgue space $L^2((0,+\infty), M, \mu)$, and the space $\mathcal{F}$ of local stress evolutions is the weighted time-dependent Lebesgue space $L^2((0,+\infty), N, \mu)$, both with the usual metrization (see, e.~g., \cite[\S5.9.2]{Evans:1998}). As Hilbert spaces, $\mathcal{F} = \mathcal{H}^*$ and $\mathcal{H} = \mathcal{F}^*$, with Riesz mapping given by the timewise application of Hooke's law $\sigma_t(\tau) = \mathbb{C} \epsilon_t(\tau)$. 
\end{definition}

Let $(\phi_k)_{k=1}^\infty$ be an orthonormal basis of $\mathcal{H}$. Then, strain histories admit the representation
\begin{equation} \label{A6y3Fr}
    \epsilon_t(\tau)
    =
    \sum_{k=1}^\infty 
    ( \epsilon_t, \phi_k )
    \, \phi_k(\tau) .   
\end{equation}
Inserting representation (\ref{A6y3Fr}) into (\ref{yCn5JN}), gives
\begin{equation} \label{AJ3FYj} 
    S \epsilon_t(\tau) 
    = 
    \sum_{k=1}^\infty ( \epsilon_t, \phi_k ) \, \psi_k(\tau) ,
\end{equation}
where 
\begin{equation}
    \psi_k(\tau)
    =
    \int_0^{+\infty}
        \mathbb{C}^{-1} \mathbb{K}(\rho) \, \phi_k(\tau + \rho)
    \, d\rho ,
\end{equation}
provided that the series converges, in which case (\ref{AJ3FYj}) is known as a {\sl Hilbert-Schmidt representation} of $S$. 

We observe from representation (\ref{A6y3Fr}) that the variables 
\begin{equation} \label{hgBWG1}
    q_k(t) = ( \epsilon_t, \phi_k )
\end{equation}
record sufficient information to reconstruct the entire history of strain, and can therefore be regarded as {\sl history variables}. In addition, we see from (\ref{bgmUJR}) and (\ref{AJ3FYj}) that the variables $q_k(t)$, together with $\epsilon(t)$, fully characterize the instantaneous state of the material at time $t$ and, therefore, can be interpreted as {\sl internal variables}. Internal variable representations of materials with memory date back to the work of C.~Eckart \cite{Eckart:1940, Eckart:1948}, Meixner \cite{Meixner:1953}, Biot \cite{Biot:1954} and Ziegler \cite{Ziegler:1958} and were formalized further by Coleman and Gurtin \cite{Coleman:1967} and others \cite{Rice:1971, lubliner1973a, Rice:1975}. The physical foundations underlying the memory-functional and the internal-variable formalisms were critically reviewed by Kestin and Rice \cite{Kestin:1970}. We note that the set of internal variables (\ref{hgBWG1}) depends on the choice of basis and, therefore, is not unique, with different sets related by a change of basis \cite{lubliner1973a}.

\begin{example}[Laguerre polynomials] \label{eFPcp3} {\rm
Suppose the weight $w(t)$ is of the exponential form (\ref{AJN7Uj}), denoting an exponentially fading memory. We recall that the {\sl Laguerre polynomials}
\begin{equation}
    L_k(x) 
    = 
    \frac{{\rm e}^x}{k!}
    \frac{d^k}{dx^k}
    ({\rm e}^{-x} x^k) 
    =
    \sum_{p=0}^k 
    \binom{k}{p} \frac{(-1)^p}{p!} x^p,
    \quad
    k = 0,1,\dots ,
\end{equation}
define an orthonormal basis of $L^2((0,+\infty), {\rm e}^{-x} \, dx)$. Therefore, the renormalized Laguerre polynomials
\begin{equation}
    \phi_k(\tau) = \lambda_0^{1/2} \, L_{k-1}(\lambda_0 \, \tau) 
    \quad
    k = 1,2,\dots ,
\end{equation}
define an orthonormal basis in the space $L^2((0,+\infty), w(\tau) \, d\tau)$ and $e_i \, \phi_k$, with $(e_i)_{i=1}^d$, $d=(n+1)n/2$ an orthonormal basis of $E$, defines a basis in the space $\mathcal{H}$ of local strain histories. Because of the Hilbert-space structure of the spaces of local histories, the completeness and orthonormality of the Laguerre polynomials render them well-suited for representing and sampling past local histories of viscoelastic materials with exponential decay. 
} \hfill$\square$
\end{example}

\subsection{Compactness and Hilbert-Schmidt properties}

Assuming that the plastic-strain history operator $S$ admits a Hilbert-Schmidt representation (\ref{AJ3FYj}), it is natural to seek approximations of the form
\begin{equation} \label{zJ8aMM}
    S_N \epsilon_t(\tau) 
    = 
    \sum_{k=1}^N ( \epsilon_t, \phi_k )  \, \psi_k(\tau) ,
\end{equation} 
where $(\phi_k)_{k=1}^N$ and $(\psi_k)_{k=1}^N$ are functions in $\mathcal{H}$, to be determined. Operators of the form (\ref{zJ8aMM}) are said to be of {\sl finite-rank}. The central question is, then, under what conditions and in what sense the sequence of operators $(S_N)$ converges to $S$. An ancillary question, taken up in Section~\ref{r4ElCe}, concerns the optimal choice of basis functions $(\phi_k)_{k=1}^N$ and $(\psi_k)_{k=1}^N$ for given rank $N$. We note that this question is equivalent to that of determining the best set of internal variables of a given dimension, see discussion at the end of Section~\ref{cN43Ba}.

A bounded operator on a Hilbert space is said to have the {\sl approximation property} if it can be approximated by finite-rank operators of the form (\ref{zJ8aMM}) in the sense of the operator norm. A classical result of analysis (see, e.~g., \cite[Cor.~6.2.]{Brezis:2010}) is that a bounded operator has the approximation property if and only if it is {\sl compact}. An important class of compact operators is the class of {\sl Hilbert-Schmidt operators}. An application of Hilbert-Schmidt theory, see Appendix~\ref{rQw4p3}, to (\ref{yCn5JN}) yields the following compactness result.

{Henceforth, we restrict attention throughout to histories of finite duration $T>0$, and redefine the space of histories accordingly as $\mathcal{H}$ $:=$ $L^2((0,T), M, \mu)$.}

\begin{thm}[Compactness] \label{mCpR5D}
Assume:
\begin{itemize}
\item[i)] (Elastic stability). $\mathbb{C} \in L(\mathbb{R}^{{n}\times {n}}_{\rm sym})$, $\mathbb{C}^T = \mathbb{C}$, $\mathbb{C} > 0$.
\item[ii)] (Hilbert-Schmidt). There is a positive, continuous, non-increasing, square-integrable weighting function $w(\tau) : [0,+\infty) \to \mathbb{R}$, normalized to $w(0) = 1$,  satisfying the semigroup condition (\ref{ttm59H}), and $\gamma > 0$ such that
\begin{equation} \label{DcjG6H}
    \int_0^{+\infty}
        \| \mathbb{K}(\tau) \|^2
    \, w^{-1}(\tau) \, d\tau
    \leq
    \gamma^2 ,
\end{equation}
where $\| \cdot \|$ denotes the operator norm.
\end{itemize}
Then, the plastic-strain history operator $S$, eq.~(\ref{yCn5JN}), is Hilbert-Schmidt, hence compact, in $\mathcal{H}$ and
\begin{equation} \label{4tgDYk}
    \| S \|_{\rm HS} \leq \gamma {\sqrt{T}} 
\end{equation}
\end{thm}

\begin{proof} \normalsize 
{Applying the test (\ref{9dqPFC}) to $S$ over $\mathcal{H}$, we obtain
\begin{equation} 
\begin{split}
    \| S \|_{\rm HS}^2
    & =
    \int_0^T
        \Big(
            \int_\tau^T
                \| \mathbb{K}(\rho-\tau) \|^2
                w^{-1}(\rho)
            \, d\rho
        \Big)
    w(\tau) \, d\tau .
\end{split}
\end{equation}
Effecting the change of variables $\theta = \rho - \tau$, we equivalently have
\begin{equation} 
\begin{split}
    \| S \|_{\rm HS}^2
    & =
    \int_0^T
        \Big(
            \int_0^{T-\tau}
                \| \mathbb{K}(\theta) \|^2
                w^{-1}(\theta+\tau)
            \, d\theta
        \Big)
    w(\tau) \, d\tau .
\end{split}
\end{equation}
Finally, using the semigroup property (\ref{ttm59H}), assumption (ii), and saturating the bound, we find
\begin{equation} 
\begin{split}
    \| S \|_{\rm HS}^2
    & \leq
    \int_0^T
        \Big(
            \int_0^\infty
                \| \mathbb{K}(\theta) \|^2
                w^{-1}(\theta)
            \, d\theta
        \Big)
    \, d\tau 
    \leq
    \gamma^2 T ,
\end{split}
\end{equation}
whence (\ref{4tgDYk}) follows.} 
\end{proof}

{Remarkably, compactness hinges critically on the boundedness of the time domain, as expected from Volterra operators. In particular, the Hilbert-Schmidt and compactness properties of $S$ fail if $T=+\infty$.}

In the special case of $\mathbb{K}(\tau)$ essentially bounded, an application of H\"older's inequality gives
\begin{equation} 
    \int_0^{+\infty}
        \| \mathbb{K}(\tau) \|^2
    \, w^{-1}(\tau) \, d\tau
    \leq
    \| \mathbb{K} \|_\infty
    \int_0^{+\infty}
        \| \mathbb{K}(\tau) \|
    \, w^{-1}(\tau) \, d\tau ,
\end{equation}
{and the condition of fading memory (\ref{jVy9Vz}) implies the Hilbert-Schmidt condition (\ref{DcjG6H}).} 
For instance, suppose that the weight $w(t)$ is of the form (\ref{AJN7Uj}) and $\mathbb{K}(\tau) \sim {\rm e}^{-\lambda \tau}$. Then, the left-hand side of (\ref{jVy9Vz}) evaluates to $\sim 1/(\lambda-\lambda_0)$, which requires that $\lambda > \lambda_0$; whereas the left-hand side of (\ref{DcjG6H}) evaluates to $\sim 1/(2\lambda-\lambda_0)$, which requires the weaker condition $\lambda > \lambda_0/2$.

\subsection{Best rank-$N$ approximation of a history operator} \label{r4ElCe}

Suppose that the plastic-strain history operator $S$ is compact, see Theorem~\ref{mCpR5D}. Then, by the approximation property, it follows that $S$ can be approximated by finite-rank operators of the form (\ref{zJ8aMM}) in the operator norm. However, this approximation property begs the question of which is the best approximation of a given rank, i.~e., the rank-$N$ operator $S_N$ such that 
\begin{equation} \label{7rPf4y}
    \| S - S_N \| = \inf_{\operatorname{\rm rank} T \leq N} \| S - T \| ,
\end{equation}
if any. For Hilbert-Schmidt operators, this problem was fist formulated and solved by Schmidt in 1907 \cite{Schmidt:1907}. We observe that the optimality of the operator $S_N$, in the sense (\ref{7rPf4y}), implies an optimal error bound (\ref{8dcSfH}) for the corresponding histories.

Remarkably, problem (\ref{7rPf4y}) falls squarely within the theory of $N$-widths \cite{Pinkus:1985}. Appendix~\ref{NyP5uW} summarizes some basic aspects of the theory for ease of reference. In the present setting, a central result of the theory, Theorem~\ref{Mj6553}, is that the optimal rank-$N$ approximation $S_N$ of the plastic-strain history operator $S$ can be characterized in terms of eigenvalues and eigenfunctions of the operators ${S}^*{S}$ and ${S}{S}^*$, where \begin{equation}
    S^*\epsilon_t(\tau)
    =
    \int_0^\tau
        \mathbb{C}^{-1}
        \mathbb{K}(\tau-\rho) \, \epsilon_t(\rho)
        \, \frac{w(\rho)}{w(\tau)} 
    \, d\rho ,
\end{equation}
is the {\sl adjoint} plastic-strain operator in $\mathcal{M}$.

If $S$ is compact, $S^*S$ and $S S^*$ are compact and self-adjoint operators. It then follows that $S^*S$ and $S S^*$ define a sequence of positive real eigenvalues $(\mu_k)_{k=1}^{\operatorname{\rm rank}S}$ such that the sequence is non-increasing and, if $\operatorname{\rm rank}S=\infty$, $\lim_{k\to\infty} \mu_k = 0$. In addition, with $N \leq \operatorname{\rm rank} S$ the eigenvectors $(\phi_k)_{k=1}^N$ of ${S}^*{S}$ and the eigenvectors $(\psi_k)_{k=1}^N$ of ${S}{S}^*$, related by $\psi_k = {S} \phi_k$, define the best rank-$N$ approximation (\ref{zJ8aMM}) of ${S}$, in the sense of (\ref{7rPf4y}).

\begin{example}[Standard linear-isotropic solid] \label{PBoGIl} {\rm 
Suppose, as in Example~\ref{eFPcp3}, that the weight $w(t)$ is of the exponential form (\ref{AJN7Uj}), corresponding to exponentially fading memory. Consider a material known to be isotropic. Then a joint modal decomposition of $\mathbb{C}$ and $\mathbb{K}(\tau)$ decouples the hereditary law into $d=n(n+1)/2$ independent scalar hereditary laws. Suppose that, for each mode, the modal elastic moduli and relaxation kernel are of the standard linear solid form
\begin{equation} \label{LASldN}
    \mathbb{C} = C_0 + C_1 ,
    \quad
    \mathbb{K}(\tau)
    =
    C_1 \lambda_1 \, {\rm e}^{-\lambda_1 \tau} ,
\end{equation}
with $\lambda_1 > \lambda_0$, which is a special case of the Maxwell-Wiechert model, \S\ref{t3QNXc}, with $N=1$. Suppose that the strain histories of interest are supported in the interval $[0,T]$. Then, a simple 
calculation shows that the functions
\begin{equation} \label{6jblrZ}
    \phi_k(\tau)
    =
    c_k \, {\rm e}^{-\nicefrac{1}{2}(\lambda_1-\lambda_0) \tau}
    \sin(\omega_k\tau) ,
    \quad
    k=1,2,\dots ,
\end{equation}
with
\begin{equation}
    c_k
    =
    \Big(\frac{2\omega_k^2}{\lambda_1(\lambda_1^2+4\omega_k^2)}\Big)^{-1/2} ,
    \quad
    \omega_k = \frac{k\pi}{T} ,
\end{equation}
are orthonormal eigenfunctions of $S^*S$ with eigenvalues
\begin{equation} \label{BmZWu1}
    \mu_k
    =
    \Big( \frac{C_1}{C_0+C_1}\Big)^2 \,
    \frac{4 \lambda_1^2}{(\lambda_1-\lambda_0)^2+\omega_k^2} ,
\end{equation}
provided that $\lambda_0 < \lambda_1$. It then follows from Theorem~\ref{Mj6553} that the functions $\phi_k(\tau)$, eq.~(\ref{6jblrZ}), and the corresponding functions $\psi_k(\tau) = S \phi_k(\tau)$, supply the most accurate finite-rank representations, eq.~(\ref{zJ8aMM}), of the plastic-strain history operator for the standard-linear solid. Correspondingly, for a given number of variables, the history variables (\ref{hgBWG1}) contain the most information pertaining to the past history of strain. 
}\hfill$\square$
\end{example}

\subsection{Best rank-$N$ approximation of a class of history operators} 

Suppose now that we wish to identify plastic-strain history subspaces of $\mathcal{H}$ that afford the best possible representation for an entire {\sl class} $\mathcal{S}$ of compact plastic-strain history operators. For instance, whereas the plastic-strain history operator of a particular material may not be known, it may be possible to surmise some of its properties {\sl a priori}, e.~g., bounds on decay rates, which defines classes of operators to which the unknown operator is sure to belong. 

This problem again falls within the framework of $N$-widths \cite{Pinkus:1985}. Specifically, we wish to minimize the $N$-width of the entire class $\mathcal{S}$, which is now given by the largest $s$-number $s_{N+1}({T})$ among all plastic-strain operators in the class, namely,
\begin{equation}
\begin{split}
    &
    d_N(\mathcal{S} \, B_1(\mathcal{H}))
    =
    \sup_{{T} \in \mathcal{S}} \mu_{N+1}^{1/2}({T}{T}^*) 
    = 
    \sup_{{T} \in \mathcal{S}} 
    \left[
    \inf_{\operatorname{\rm dim} {V} \leq N}
    \sup_{( \eta,\, {V})=0}
    \frac{( {T}{T}^*\eta,\, \eta)}{( \eta, \eta)} 
    \right]^{1/2} ,
\end{split}
\end{equation}
in terms of the Rayleigh-Ritz characterization of the $(N+1)$ eigenvalue of ${T}{T}^*$, see (\ref{gp73Mt}). 

Suppose, specifically, that there is a bounding Hilbert-Schmidt operator ${S}$ such that
\begin{equation} \label{qgIp90}
    \mathcal{S}
    =
    \{
        {T} \in \mathcal{B}(\mathcal{H}) \, : \, 
        ({T} \eta, {T} \eta) \leq ({S} \eta, {S} \eta) , \;
        \forall \eta \in \mathcal{H}
    \} .
\end{equation}
By (\ref{3xGVzH}), it follows immediately that every ${T} \in \mathcal{S}$ is Hilbert-Schmidt, hence compact. Alternatively, let $(\mu_k)_{k=1}^{\operatorname{\rm rank}S}$, $(\phi_k)_{k=1}^{\operatorname{\rm rank}S}$ and $(\psi_k)_{k=1}^{\operatorname{\rm rank}S}$, with $\psi_k = {S} \phi_k$, define the best rank-$N$ approximation of ${S}$, in the sense of $N$-widths. Then,
\begin{equation} \label{iHjYDY}
    \mathcal{S}
    =
    \{
        {T} \in \mathcal{B}(\mathcal{H}) \, : \, 
        ({T} \phi_k,\, {T} \phi_k) \leq \mu_k , \;
        k=1,\dots,\operatorname{\rm rank}S
    \} ,
\end{equation}
which characterizes the class of plastic-strain history operators in terms of the Hilbert-Schmidt representation of the bounding operator. 

It is readily verified from (\ref{qgIp90}) that the set $\mathcal{S}$ is convex. Suppose, in addition that $\mathcal{S}$ is a compact subset of $\mathcal{B}(\mathcal{H})$. For instance, we may consider operators with hereditary kernel of the form (\ref{eqZ2Sh}) in terms of finite spectral measures with compact support. Then, by the minimax theorem \cite{Sion:1958} we have
\begin{equation} \label{aLTm6N}
    d_N(\mathcal{S} \, B_1(\mathcal{H}))
    =
    \left[
    \inf_{\operatorname{\rm dim} {V} \leq N}
    \sup_{( \eta,\, {V})=0}
    \sup_{{T} \in \mathcal{S}} 
    \frac{( {T}\eta,\, {T}\eta)}{( \eta, \eta)} 
    \right]^{1/2} ,
\end{equation}
In addition, ${S} \in \mathcal{S}$ and 
\begin{equation} \label{BbpV4j}
    \sup_{{T} \in \mathcal{S}} 
    ( {T}\eta,\, {T}\eta)
    =
    ({S} \eta, {S} \eta) , 
    \quad
    \forall \eta \in \mathcal{H} ,
\end{equation}
whence it follows that the representations
\begin{equation} 
    {T}_N \eta 
    = 
    \sum_{k=1}^N c_k ( \eta, \phi_k ) \, \phi_k ,
    \quad
    0 \leq c_k \leq \mu_k ,
    \quad
    N = 1,\dots, \operatorname{\rm rank}{S} ,
\end{equation} 
are optimal for the entire set $\mathcal{S}$ in the sense of $d_N(\mathcal{S} \, B_1(\mathcal{H}))$.

\begin{example}[Bounding standard linear-isotropic solid] {\rm Suppose that the standard linear-isotropic solid of Example~\ref{PBoGIl} is known to supply a bounding plastic-strain history operator ${S}$ for a class $\mathcal{S}$ of materials, in the sense of (\ref{qgIp90}). Suppose that ${T}$ is of the form (\ref{yCn5JN}) with spectral representation (\ref{eqZ2Sh}). A straightforward calculation using (\ref{6jblrZ}) gives
\begin{equation}
    {T} \phi_k
    =
    \frac{c_k}{C_0+C_1} \, 
    {\rm e}^{-\frac{1}{2} (\lambda_1-\lambda_0)\tau  } 
    (A_k \, \cos(\omega_k\tau) + B_k \, \sin(\omega_k\tau))
\end{equation}
where
\begin{subequations}
\begin{align}
    &
    A_k
    =
    \int_{\lambda_0}^{+\infty}
        \frac
        {
            4 \omega_k 
        }
        {
            (2 \lambda + \lambda_1 - \lambda_0)^2+4 \omega_k ^2
        } 
        (\lambda-\lambda_0) 
    \, d\nu(\lambda) ,
    \\ &
    B_k
    =
    \int_{\lambda_0}^{+\infty}
        \frac
        {
            2 (2 \lambda + \lambda_1 - \lambda_0)
        }
        {
            (2 \lambda + \lambda_1 - \lambda_0)^2+4 \omega_k^2
        }
        (\lambda-\lambda_0) \, 
    \, d\nu(\lambda) ,
\end{align}
\end{subequations}
and
\begin{equation} \label{PmBYnH}
    ( T \phi_k,\, T\phi_k)
    =
    \frac{c_k^2}{C_0+C_1}
    \frac
    {
        2 \omega_k^2 
        \left(A_k^2+B_k^2\right)
        +
        A_k^2 \lambda_1^2
        +
        2 A_k B_k
        \lambda_1 \omega_k 
    }
    {
        \lambda_1(\lambda_1^2+4 \omega_k^2)
    } .
\end{equation}
Then, ${T}$ is in $\mathcal{S}$ if conditions (\ref{iHjYDY}) are satisfied, with $\mu_k$ as in (\ref{BmZWu1}). We note that, by the Krein-Milman theorem \cite[3.23]{Rudin:1991}, the conditions (\ref{PmBYnH}) need only verified for Dirac spectral measures $\nu(\lambda)$. 
}\hfill$\square$
\end{example}

\section*{Acknowledgements}

The financial support of the {\sl Centre Internacional de Mètodes Numèrics a l'Enginyeria} (CIMNE) of the {\sl Universitat Politecnica de Catalunya} (UPC), Spain, through the {\sl UNESCO Chair in Numerical Methods in Engineering} is gratefully acknowledged.

\begin{appendix}

\section{Compact and Hilbert-Schmidt operators} \label{rQw4p3}

We recall that a bounded operator ${S} \in \mathcal{B}(H)$ on a Hilbert space $H$ is compact if and only if it can be approximated by finite-rank operators of the form (\ref{zJ8aMM}), with the sum converging in the operator norm. Thus, on a Hilbert space, any compact operator is a limit of finite-rank operators, i.~e., the class of compact operators is the closure of the set of finite-rank operators in the norm topology, which is referred to as the approximation property \cite[Cor.~6.2.]{Brezis:2010}, whence it also follows that compact operators are necessarily bounded (see, e.~g.,  \cite[\S4.16]{Rudin:1991}). Thus, it follows that finite-rank approximation of the form (\ref{zJ8aMM}) converge, in the operator norm, if and only if $P$ is compact.

An important class of compact operators is the class of {\sl Hilbert-Schmidt operators}. We recall that ${S} \in \mathcal{B}(H)$ is a Hilbert-Schmidt operator if \cite[Ex.~IX.2.19]{Conway:1990}
\begin{equation} \label{3xGVzH}
    \| {S} \|_{\rm HS} 
    := 
    \Big( \sum_{k=1}^\infty \| {S} e_k \|^2 \Big)^{1/2}
    < +\infty ,
\end{equation}
where $\|{S}\|_{\rm HS} \geq \|{S}\|$ is the {\sl Hilbert-Schmidt norm} of ${S}$ and $(e_i)_{k=1}^\infty$ is an orthonormal basis of $H$. It is readily verified that the definition (\ref{3xGVzH}) is independent of the choice of basis. A classical result from analysis is that Hilbert-Schmidt operators are indeed compact \cite[Ex.~4.15]{Rudin:1991}. In particular, suppose that $H = L^2((a,b),\mu)$ and ${S} \in \mathcal{B}(H)$ is an integral operator with kernel $k(s,t)$, i.~e.,
\begin{equation}
    {S}\xi(t) = \int_a^b k(t,s) \xi(s) \, d\mu(s) ,
\end{equation}
for $\xi \in H$. Suppose, in addition, that $k(s,t) \in L^2((a,b)^2,\mu\times\mu)$. Then \cite[Ex.~IX.2.19]{Conway:1990}, \cite[Ex.~4.15]{Rudin:1991},
\begin{equation} \label{9dqPFC}
    \| {S} \|_{\rm HS}^2
    =
    \int_a^b \int_a^b 
        | k(t,s) |^2
    \, d\mu(t) \,d\mu(s) 
    <
    +\infty ,
\end{equation}
and ${S}$ is Hilbert-Schmidt. 

\section{$N$-widths}
\label{NyP5uW}

Recall that the Kolmogoroff $N$-width of a subset $A$ of a Hilbert space $H$ is \cite[Def.~\S I.1]{Pinkus:1985}
\begin{equation} \label{c6fx6F}
    d_N(A; H) 
    = 
    \inf_{\operatorname{\rm dim} V = N} \,
    \sup_{\xi \in A} \inf_{\eta \in {V}} \| \xi - \eta \| ,
\end{equation}
where the infimum is taken over all $N$-dimensional subspaces ${V}$ of $H$. The $N$-width of $A$ measures how well $A$ can be approximated by $N$-dimensional subspaces of $H$. If
\begin{equation} \label{9PJ7Wx}
    d_N(A; H) = \sup_{\xi \in A} \inf_{\eta \in H_N} \| \xi - \eta \| 
\end{equation}
for some subspace $H_N$ of dimension at most $N$, then $H_N$ is an optimal subspace for $d_N (A; X)$. Since 
\begin{equation} 
\begin{split}
    &
    \inf_{\operatorname{\rm rank} S \leq N}\| {T} - S\| 
    = 
    \inf_{\operatorname{\rm rank} S \leq N}
    \sup_{\|\xi\|\leq 1} 
    \|({T} - S)\xi\|
    = \\ &
    \inf_{\operatorname{\rm dim} V = N} \,
    \sup_{\|\xi\|\leq 1} 
    \inf_{\eta\in {V}}
    \|{T} \xi - \eta\|
    =
    d_N({T}B_1(H)) ,
\end{split}
\end{equation}
where $B_1(H) = \{ \xi \in H \, : \, \| \xi \| \leq 1 \}$ is the unit ball of $H$ and 
\begin{equation}
    {T} B_1(H)
    = 
    \{ {T} \xi \, : \, \| \xi \| \leq 1 \} ,
\end{equation}
it follows that the best-approximation criterion (\ref{7rPf4y}) is equivalent to minimizing the Kolmogoroff $N$-width of ${T}B_1(H)$. If
\begin{equation} \label{Y4BfXS}
    \inf_{\operatorname{\rm rank} S \leq N}\| {T} - S\| 
    =
    \| {T} - {T}_N \| 
\end{equation}
for some operator ${T}_N$ of rank at most $N$, then ${T}_N$ is an optimal operator for $d_N (A; X)$.

In Hilbert spaces, Kolmogoroff $N$-widths can be characterized in terms of eigenvalues of the self-adjoint, non-negative compact operators ${T}^*{T}$ and ${T}{T}^*$, where ${T}^*$ is the adjoint operator with the property that
\begin{equation}
    ({T}\xi,\eta) = (\xi,{T}^*\eta) ,
\end{equation}
for all $\xi$, $\eta \in H$. We recall \cite[Theorem 2.1]{Pinkus:1985} that, if $S\in \mathcal{B}(H)$ is a compact, non-negative self-adjoint operator, then there is a sequence of non-zero real eigenvalues $(\mu_k)_{k=1}^N$, with $N$ the rank of $S$, such that the sequence $(|\mu_k|)_{k=1}^N$ is non-increasing and, if $N=\infty$, $\lim_{i\to\infty} \mu_k = 0$. Furthermore, if every eigenvalue is repeated in the sequence according to its multiplicity, then there exists a sequence $(\phi_k)_{k=1}^N$ of corresponding eigenfunctions that form an orthonormal basis for the range of $S$. An appeal to the classical duality representation  
\begin{equation} \label{xrlrGA}
    \inf_{\eta\in {V}}
    \|{T} \xi - \eta\| 
    =
    \sup_{( \eta,\, {V})=0}
    \frac{( {T} \xi, \eta)}{\|\eta\|} ,
\end{equation}
then gives
\begin{equation} \label{gp73Mt}
\begin{split}
    &
    d_N({T}B_1(H))
    =
    \inf_{\operatorname{\rm dim} {V} \leq N}
    \sup_{\|\xi\|\leq 1} 
    \sup_{( \eta,\, {V})=0}
    \frac{( {T} \xi, \eta)}{\|\eta\|} 
    = \\ &
    \inf_{\operatorname{\rm dim} {V} \leq N}
    \sup_{( \eta,\, {V})=0}
    \sup_{\|\xi\|\leq 1} 
    \frac{( \xi, {T}^* \eta)}{\|\eta\|} 
    = 
    \inf_{\operatorname{\rm dim} {V} \leq N}
    \sup_{( \eta,\, {V})=0}
    \frac{\| {T}^*\eta\|}{\|\eta\|} 
    = \\ &
    \left[
    \inf_{\operatorname{\rm dim} {V} \leq N}
    \sup_{( \eta,\, {V})=0}
    \frac{( {T}{T}^*\eta,\, \eta)}{( \eta, \eta)} 
    \right]^{1/2}
    =
    \mu_{N+1}^{1/2}({T}{T}^*) ,
\end{split}
\end{equation}
which is the Rayleigh-Ritz characterization of the $(N+1)$ eigenvalue of ${T}{T}^*$. 

We recall that the $s$-numbers, or {\sl singular values}, of ${T}$ \cite[Chapter IV]{Pinkus:1985}, first introduced by E.~Schmidt \cite{Schmidt:1907}, are the numbers
\begin{equation}
    s_k({T})
    =
    [\mu_k({T}^*{T})]^{1/2},
    \quad
    i = 1,\dots,\operatorname{\rm rank} {T}^*{T} ,
\end{equation}
where $(\mu_k({T}^*{T}))_{k=1}^N$ are the eigenvalues of ${T}^*{T}$, repeated according to their multiplicity and we order the eigenvalues so that
\begin{equation}
    \mu_1({T}^*{T}) \geq \mu_2({T}^*{T}) \geq \dots > 0 .
\end{equation}
It is readily verified that $s_k ({T}) = s_k ({T}^*)$. Denote by $(\phi_k)_{k=1}^N$ the orthonormal eigenvectors of ${T}^*{T}$ and set $\psi_k = {T} \, \phi_k$. It follows that $(\psi_k)_{k=1}^N$ are eigenvectors of ${T}{T}^*$, with the orthogonality property
\begin{equation}
    ( \psi_k ,\, \psi_l )
    =
    s_l^2({T}) \, \delta_{kl} .
\end{equation}
Then, the following theorem characterizes the sought optimal finite-rank approximation of a compact operator ${T}$. 

\begin{thm} {\rm (Adapted from \cite[Theorem \S IV 2.2]{Pinkus:1985})}. \label{Mj6553}
Let ${T} \in \mathcal{B}(H)$, compact. For $N=0,1,\dots,\operatorname{\rm rank} {T}^*{T}$, let $(s_k({T}))_{k=1}^N$, $(\phi_k)_{k=1}^N$ and $(\psi_k)_{k=1}^N$ be the singular values of ${T}$, orthonormal eigenvectors of ${T}^*{T}$ and eigenvectors of ${T}{T}^*$, respectively. Then:
\begin{itemize}
\item[a)] $d_N({T}B_1(H); H) = s_{N+1}({T})$.
\item[b)] The subspaces $H_N = \{\psi_1,\dots, \psi_N\}$ are optimal in the sense (\ref{9PJ7Wx}).
\item[c)] The linear operators  
\begin{equation} 
    {T}_N \eta = \sum_{k=1}^N ( \eta, \phi_k ) \, \psi_k ,
\end{equation} 
are optimal in the sense (\ref{Y4BfXS}).
\end{itemize}
\end{thm}

\end{appendix}

\bibliographystyle{unsrt}

\begin{thebibliography}{100}

\bibitem{markovitz1977a}
H.~Markovitz.
\newblock Boltzmann and the beginnings of linear viscoelasticity.
\newblock {\em Trans. Soc. Rheol.}, 23:381--398, 1977.

\bibitem{Golden:1989}
J.~M. Golden and G.~A.~C. Graham.
\newblock {\em Boundary Value Problems in Linear Viscoelasticity}.
\newblock Springer-Verlag, Berlin Heidelberg, 1988.

\bibitem{weber1835a}
W.~Weber.
\newblock Die elastizitat der seidenfaden.
\newblock {\em Ann. Phys. Chem.}, 34:247--257, 1835.

\bibitem{weber1841a}
W.~Weber.
\newblock {\"U}ber die elastizitat fester korper.
\newblock {\em Ann. Phys. Chem.}, 54:1--18, 1841.

\bibitem{meyer1874a}
O.E. Meyer.
\newblock Theorie der elastischen nachwirkung.
\newblock {\em Ann. Phys. Chem.}, CSI:108–118, 1874.

\bibitem{Boltzmann:1874}
L.~Boltzmann.
\newblock Zur theorie der elastischen nachwirkungen.
\newblock {\em Sitzungsberichte der Mathematisch-Naturwissenschaftlichen Classe
  der Kaiseflichen Akademie der Wissenschaffen}, 70(2):275--306, 1874.

\bibitem{boltzmann1878a}
L.~Boltzmann.
\newblock Zur theorie der elastischen nachwirkung.
\newblock {\em Ann. Phys. Chem.}, 5:430–432, 1878.

\bibitem{voigt1890a}
Woldermar Voigt.
\newblock {\"U}ber die innere reibung der festen k\"orper, insbesondere der
  krystalle.
\newblock In {\em Abhandlungen der K\"oniglichen Gesellschaft von
  Wissenschaften zu G\"ottingen}, volume~36, page 3–47. 1890.

\bibitem{Wiechert:1893}
E.~Wiechert.
\newblock Gesetze der elastischen nachwirkung f\"ur constante temperatur.
\newblock {\em Ann. Phys.}, 286(11):546–--570, 1893.

\bibitem{Wiechert:1899}
E.~Wiechert.
\newblock {\em {\"U}ber elastische Nachwirkung}.
\newblock PhD thesis, K\"onigsberg University, Germany, 1889.

\bibitem{thomson:1865a}
W.~Thomson.
\newblock {IV.~O}n the elasticity and viscosity of metals.
\newblock {\em Proc. R. Soc. Lond.}, 14:289--297, 1865.

\bibitem{maxwell1867a}
J.C. Maxwell.
\newblock On the dynamical theory of gases.
\newblock {\em Philos. Trans. R. Soc. Lond.}, 157:49--88, 1867.

\bibitem{love1892a}
A.E.H. Love.
\newblock {\em A Treatise on the Mathematical Theory of Elasticity}.
\newblock Cambridge University Press, Cambridge, 1892.

\bibitem{alfrey1948a}
T.~Jr. Alfrey.
\newblock {\em Mechanical Behavior of High Polymers}.
\newblock Interscience, New York, 1948.

\bibitem{eirich1956a}
F.R. Eirich.
\newblock {\em Rheology Theory and Applications}.
\newblock Academic, New York, 1956.

\bibitem{staverman1956a}
A.J. Staverman and E.~Schwarzl.
\newblock Linear deformation behaviour of high polymers.
\newblock In H.A. Stuart, editor, {\em Die Physik der Hochpolymeren},
  volume~IV. Springer, Berlin, 1956.

\bibitem{bingham1922}
E.~Bingham.
\newblock {\em Fluidity and Plasticity}.
\newblock McGraw-Hill, New York, 1922.

\bibitem{ferry1961}
J.~Ferry.
\newblock {\em Viscoelastic Properties of Polymers}.
\newblock John Wiley \& Sons, New York, 1961.

\bibitem{Rivlin:1951}
R.~S. Rivlin and D.~W. Saunders.
\newblock Large elastic deformations of isotropic materials vii. experiments on
  the deformation of rubber.
\newblock {\em Philos. Trans. R. Soc. A}, 243(865):251--288, 1951.

\bibitem{rivlin1955a}
R.S. Rivlin and J.L. Ericksen.
\newblock Stress relaxation for isotropic materials.
\newblock {\em J. Rat. Mech. Anal.}, 4:323--425, 1955.

\bibitem{rivlin1965a}
R.S. Rivlin.
\newblock Nonlinear viscoelastic solids.
\newblock {\em SIAM Review}, 7:323--340, 1965.

\bibitem{oldroyd1950}
J.G. Oldroyd.
\newblock On the formulation of rheological equations of state.
\newblock {\em Proc. R. Soc. Lond. Ser. A}, 200(1063):523--541, 1950.

\bibitem{schapery1969a}
R.A. Schapery.
\newblock On the characterization of nonlinear viscoelastic solids.
\newblock {\em Polym. Eng. Sci}, 9:295--310, 1969.

\bibitem{schapery1997a}
R.A. Schapery.
\newblock Nonlinear viscoelastic and viscoplastic constitutive equations based
  on thermodynamics.
\newblock {\em Mech. Time-Depend. Mater.}, 1(2):209--240, 1997.

\bibitem{schapery2000a}
R.A. Schapery.
\newblock Nonlinear viscoelastic solids.
\newblock {\em Int. J. Solids. Struct}, 37:359–366, 2000.

\bibitem{Menard:2002}
K.P. Menard and N.R. Menard.
\newblock {\em Dynamic Mechanical Analysis}.
\newblock CRC Press, Boca Raton, 2002.

\bibitem{Herbert:2008}
E.G. Herbert, W.C. Oliver, and G.M. Pharr.
\newblock Nanoindentation and the dynamic characterization of viscoelastic
  solids.
\newblock {\em J. Phys. D}, 41(7):074021, 2008.

\bibitem{Herbert:2009}
E.G. Herbert, W.C. Oliver, A.~Lumsdaine, and G.M. Pharr.
\newblock Measuring the constitutive behavior of viscoelastic solids in the
  time and frequency domain using flat punch nanoindentation.
\newblock {\em J. Mater. Res.}, 24(3):626--637, 2009.

\bibitem{Arbogast:1998}
K.B. Arbogast and S.S. Margulies.
\newblock Material characterization of the brainstem from oscillatory shear
  tests.
\newblock {\em J. Biomech.}, 31(9):801--807, 1998.

\bibitem{Bayly:2008}
P.V. Bayly, P.G. Massouros, E.~Christoforou, A.~Sabet, and G.M. Genin.
\newblock Magnetic resonance measurement of transient shear wave propagation in
  a viscoelastic gel cylinder.
\newblock {\em J. Mech. Phys. Solids}, 56(5):2036--2049, 2008.

\bibitem{Wilhelm:1998}
M.~Wilhelm, D.~Maring, and HW. Spiess.
\newblock Fourier-transform rheology.
\newblock {\em Rheol. Acta}, 37:399--405, 1998.

\bibitem{gemant1936}
A.S. Gemant.
\newblock A concept of viscosity.
\newblock {\em Physics}, 7:311--317, 1936.

\bibitem{Fluegge:1975}
W.~Fl{\"u}gge.
\newblock {\em Viscoelasticity}.
\newblock Springer--Verlag, New York, Heidelberg, Berlin, 2nd edition, 1975.

\bibitem{tschoegl1989}
N.W. Tschoegl.
\newblock {\em The Phenomenological Theory of Linear Viscoelastic Behavior: An
  Introduction}.
\newblock Springer, New York, 1989.

\bibitem{Volterra:1909}
V.~Volterra.
\newblock Sulle equazioni integro-differenziali della teoria
  dell’elasticit{\'a}.
\newblock {\em Att. Real. Act. Lint}, 18:295, 1909.

\bibitem{Fabrizio2012}
M.~Fabrizio.
\newblock {\em Mathematicians in Bologna 1861--1960}, chapter Dario Graffi in a
  Complex Historical Period, pages 179--195.
\newblock Springer, Basel, 2012.

\bibitem{Linz:1985}
P.~Linz.
\newblock {\em Analytical and Numerical Methods for Volterra Equations}.
\newblock Studies in Applied and Numerical Mathematics. Society for Industrial
  and Applied Mathematics, Philadelphia, 1985.

\bibitem{green1957a}
A.E. Green and R.S. Rivlin.
\newblock The mechanics of non-linear materials with memory.
\newblock {\em Arch. Ration. Mech. Anal.}, 1:1--21, 1957.

\bibitem{green1959a}
A.E. Green, R.S. Rivlin, and A.J.M. Spencer.
\newblock The mechanics of materials with memory, ii.
\newblock {\em Arch. Ration. Mech. Anal.}, 3:82--90,, 1959.

\bibitem{green1960a}
A.E. Green and R.S. Rivlin.
\newblock The mechanics of materials with memory, iii.
\newblock {\em Arch. Ration. Mech. Anal.}, 4:387--404,, 1960.

\bibitem{Pipkin:1968}
A.~C. Pipkin and T.~G. Rogers.
\newblock A non-linear integral representation for viscoelastic behaviour.
\newblock {\em J. Mech. Phys. Solids}, 16:59--72, 1968.

\bibitem{pipkin1961a}
A.C. Pipkin and R.S. Rivlin.
\newblock Small deformations superposed on large deformations in materials with
  fading memory.
\newblock {\em Arch. Ration. Mech. Anal.}, pages 297--308,, 1961.

\bibitem{pipkin1964a}
A.C. Pipkin.
\newblock Small finite deformations of viscoelastic solids".
\newblock {\em Rev. Mod. Phys.}, 36:1034--1041,, 1964.

\bibitem{Truesdell:1965}
C.~Truesdell and W.~Noll.
\newblock {\em The nonlinear field theories of mechanics, Encyclopedia of
  Physics}.
\newblock Springer-Verlag, Berlin Heidelberg New York, III/S, 1965.

\bibitem{wang1965b}
C.C. Wang.
\newblock The principle of fading memory".
\newblock {\em Arch. Ration. Mech. Anal.}, 18:343–366,, 1965.

\bibitem{perzyna1967a}
P.~Perzyna.
\newblock On fading memory of a material.
\newblock {\em Arch. Mech.}, 19:537--547,, 1967.

\bibitem{lubliner1969a}
J.~Lubliner.
\newblock On fading memory in materials of evolutionary type.
\newblock {\em Acta Mech.}, 8:75--81,, 1969.

\bibitem{coleman1960a}
B.D. Coleman and W.~Noll.
\newblock An approximation theorem for functionals, with applications in
  continuum mechanics.
\newblock {\em Arch. Ration. Mech. Anal.}, 6:355--370,, 1960.

\bibitem{coleman1961a}
B.D. Coleman and W.~Noll.
\newblock Recent results in the continuum theory of viscoelastic fluids.
\newblock {\em Ann. N.Y. Acad. Sci.}, 89:672--714,, 1961.

\bibitem{coleman1961b}
B.D. Coleman and W.~Noll.
\newblock Foundations of linear viscoelasticity.
\newblock {\em Rev. Mod. Phys}, 33(2):239--249, 1961.

\bibitem{coleman1962a}
B.D. Coleman and W.~Noll.
\newblock Simple fluids with fading memory.
\newblock Symposium on Second-Order Effects in Elasticity, April 1962, Haifa,
  pages 530--552, Oxford, 1964. Pergamon Press.

\bibitem{coleman1966a}
B.D. Coleman and V.J. Mizel.
\newblock Norms and semigroups in the theory of fading memory".
\newblock {\em Arch. Ration. Mech. Anal.}, 23:87–123,, 1966.

\bibitem{wang1965a}
C.C. Wang.
\newblock The principle of fading memory.
\newblock {\em Arch. Ration. Mech. Anal.}, 18:343--366, 1965.

\bibitem{Sutton:2009}
M.~Sutton, J.J. Orteu, and H.~Schreier.
\newblock {\em Image Correlation for Shape, Motion and Deformation
  Measurements. Basic Concepts, Theory and Applications}.
\newblock Springer, New York, 2009.

\bibitem{Buljac:2018}
A.~Buljac, C.~Jailin, A.~Mendoza, J.~Neggers, T.~Taillandier-Thomas,
  A.~Bouterf, B.~Smaniotto, F.~Hild, and S.~Roux.
\newblock Digital volume correlation: Review of progress and challenges.
\newblock {\em Exp. Mech.}, 58:661--708, 06 2018.

\bibitem{Schleder:2019}
G.R. Schleder, Padilha A.C.M., C.M. Acosta, M.~Costa, and A.~Fazzio.
\newblock From {DFT} to machine learning: recent approaches to materials
  science--a review.
\newblock {\em J. Phys. Mater.}, 2(3):032001, may 2019.

\bibitem{Bernier:2020}
J.~V. Bernier, R.~M. Suter, A.~D. Rollett, and J.~D. Almer.
\newblock High-energy x-ray diffraction microscopy in materials science.
\newblock {\em Annu. Rev. Mater. Res.}, 50:395--436, 2020.

\bibitem{Wang:2024}
Z.~Wang, S.~Das, A.~Joshi, A.J.D Shaikeea, and V.S. Deshpande.
\newblock {3D} observations provide striking findings in rubber elasticity.
\newblock {\em Proc. Natl. Acad. Sci.}, 121(24):e2404205121, 2024.

\bibitem{Li:2019}
X.~Li, C.C. Roth, and D.~Mohr.
\newblock Machine-learning based temperature- and rate-dependent plasticity
  model: Application to analysis of fracture experiments on {DP} steel.
\newblock {\em Int. J. Plast.}, 118:320--344, 2019.

\bibitem{Jin:2022}
H.~Jin, T.~Jiao, R.J. Clifton, and K.S. Kim.
\newblock Dynamic fracture of a bicontinuously nanostructured copolymer: A
  deep-learning analysis of big-data-generating experiment.
\newblock {\em J. Mech. Phys. Solids}, 164:104898, 2022.

\bibitem{Eggersmann:2019}
R.~Eggersmann, T.~Kirchdoerfer, S.~Reese, L.~Stainier, and M.~Ortiz.
\newblock Model-free data-driven inelasticity.
\newblock {\em Comput. Methods Appl. Mech. Eng.}, 350:81--99, 2019.

\bibitem{Salahshoor:2023}
H.~Salahshoor and M.~Ortiz.
\newblock Model-free data-driven viscoelasticity in the frequency domain.
\newblock {\em Comput. Methods Appl. Mech. Eng.}, 403:115657, 2023.

\bibitem{Liu:2022}
B.~Liu, N.~Kovachki, Z.~Li, K.~Azizzadenesheli, A.~Anandkumar, A.M. Stuart, and
  K.~Bhattacharya.
\newblock A learning-based multiscale method and its application to inelastic
  impact problems.
\newblock {\em J. Mech. Phys. Solids}, 158:104668, 2022.

\bibitem{Bhattacharya:2023}
K.~Bhattacharya, B.~Liu, A.M. Stuart, and M.~Trautner.
\newblock Learning markovian homogenized models in viscoelasticity.
\newblock {\em Multiscale Model. Simul.}, 21(2):641--679, 2023.

\bibitem{Liu:2023}
B.~Liu, E.~Ocegueda, M.~Trautner, A.M. Stuart, and K.~Bhattacharya.
\newblock Learning macroscopic internal variables and history dependence from
  microscopic models.
\newblock {\em J. Mech. Phys. Solids}, 178:105329, 2023.

\bibitem{Weinberg:2023}
K.~Weinberg, L.~Stainier, S.~Conti, and M.~Ortiz.
\newblock Data-driven games in computational mechanics.
\newblock {\em Comput. Methods Appl. Mech. Eng.}, 417:116399, 2023.

\bibitem{Asad:2023}
F.~As’ad and C.~Farhat.
\newblock A mechanics-informed deep learning framework for data-driven
  nonlinear viscoelasticity.
\newblock {\em Comput. Methods Appl. Mech. Eng.}, 417:116463, 2023.

\bibitem{Marino:2023}
E.~Marino, M.~Flaschel, S.~Kumar, and L.~{De Lorenzis}.
\newblock Automated identification of linear viscoelastic constitutive laws
  with euclid.
\newblock {\em Mech. Mater.}, 181:104643, 2023.

\bibitem{Ghane:2024}
E.~Ghane, M.~Fagerstr{\"o}m, and M.~Mirkhalaf.
\newblock Recurrent neural networks and transfer learning for predicting
  elasto-plasticity in woven composites.
\newblock {\em Eur. J. Mech. A/Solids}, 107:105378, 2024.

\bibitem{Akerson:2024}
A.~Akerson, A.~Rajan, and K.~Bhattacharya.
\newblock Learning constitutive relations from experiments: 1. {PDE}
  constrained optimization.
\newblock {\em J. Mech. Phys. Solids}, 201:106128, 04 2025.

\bibitem{Bui:1993}
H.D. Bui.
\newblock {\em Introduction aux probl\`emes inverses en m\'ecanique des
  mat\'eriaux}.
\newblock Collection de la Direction des Etudes et Recherches d’Electricit\'e
  de France. Eds Eyrolles, Paris, 1993.

\bibitem{Prony:1795}
C.~Prony.
\newblock Essai exp\'erimental et analytique sur les lois de la dilabilit\'e
  des fluides \'elastiques, et sur celles de la force expansive de la vapeur de
  l’eau et de la vapeur de l’alkool.
\newblock {\em J. Ecole Polytechn.}, 2:24--–77, 1795.

\bibitem{Qvale:2004}
D.~Qvale and K.~Ravi-Chandar.
\newblock Viscoelastic characterization of polymers under multiaxial
  compression.
\newblock {\em Mech. Time-Depend. Mater.}, 8:193–214, 2004.

\bibitem{Knauss:2007}
W.G. Knauss and J.~Zhao.
\newblock Improved relaxation time coverage in ramp-strain histories.
\newblock {\em Mech. Time-Depend. Mater.}, 11:199--216, 2007.

\bibitem{Zhao:2007}
J.~Zhao, W.G. Knauss, and G.~Ravichandran.
\newblock Applicability of the time–temperature superposition principle in
  modeling dynamic response of a polyurea.
\newblock {\em Mech. Time-Depend. Mater.}, 11:289--308, 2007.

\bibitem{Martins:2018}
J.M.P. Martins, A.~Andrade-Campos, and S.~Thuillier.
\newblock Comparison of inverse identification strategies for constitutive
  mechanical models using full-field measurements.
\newblock {\em Int. J. Mech. Sci.}, 145:330--345, 2018.

\bibitem{Pinkus:1985}
A.~Pinkus.
\newblock {\em N-Widths in Approximation Theory}.
\newblock Ergebnisse der Mathematik und ihrer Grenzgebiete. 3. Folge / A Series
  of Modern Surveys in Mathematics. Springer-Verlag, Berlin Heidelberg, 1985.

\bibitem{Schmidt:1907}
E.~Schmidt.
\newblock Zur theorie der linearen und nichtlinearen integralgleichungen. i.
\newblock {\em Math. Ann.}, 63:433--476, 1907.

\bibitem{Kolmogoroff:1936}
A.~Kolmogoroff.
\newblock {\"U}ber die beste annäherung von funktionen einer gegebenen
  funktionenklasse.
\newblock {\em Ann. Math.}, 37(1):107--110, 1936.

\bibitem{Gelfand:1959}
I.~M. Gel'fand.
\newblock Certain problems of functional analysis and algebra (in russian).
\newblock volume~3 of {\em Proceedings of the Third All-Union Mathematical
  Congress (Moscow, 1956)}, pages 27--34, 1959.

\bibitem{Tikhomirov:1960}
V.~M. Tikhomirov.
\newblock On bernstein widths of sets in normed spaces.
\newblock {\em Dokl. Akad. Nauk SSSR}, 130:734--737, 1960.

\bibitem{Gurtin:1962}
E.~Gurtin and M.~E. Sternberg.
\newblock On the linear theory of viscoelasticity.
\newblock {\em Arch. Ration. Mech. Anal.}, 11(4):291--356, 1962.

\bibitem{fisher1973a}
M.~J. Fisher and G.~M.~C. Leitman.
\newblock The linear theory of viscoelasticity.
\newblock {\em Handbuch der Physik}, VIa/3:1--123, 1973.

\bibitem{Rudin:1991}
W.~Rudin.
\newblock {\em Functional Analysis}.
\newblock International series in pure and applied mathematics. McGraw-Hill,
  New York, second edition, 1991.

\bibitem{Evans:1998}
L.C. Evans.
\newblock {\em Partial Differential Equations}.
\newblock Graduate studies in mathematics. American Mathematical Society,
  Providence, RI, 1998.

\bibitem{AdamsFournier:2003}
Robert~A. Adams and John~J.F. Fournier.
\newblock {\em Sobolev spaces}.
\newblock Elsevier: Academic Press, Amsterdam, second edition, 2003.

\bibitem{Pipkin:1991}
Allen~C. Pipkin.
\newblock {\em A course on integral equations}.
\newblock Springer-Verlag, Berlin, Heidelberg, 1991.

\bibitem{lubliner1966a}
J.~Lubliner and J.L. Sackman.
\newblock On ageing viscoelastic materials.
\newblock {\em J. Mech. Phys. Solids}, 14:25--32,, 1966.

\bibitem{lubliner1966b}
J.~Lubliner and J.~L. Sackman.
\newblock A reciprocal theorem for an aging viscoelastic body.
\newblock {\em J. Franklin Inst.}, 282:161--170,, 1966.

\bibitem{lubliner1966c}
J.~Lubliner.
\newblock Rheological models for time-variable materials.
\newblock {\em Nucl. Eng. Des.}, 4:287--291,, 1966.

\bibitem{Coleman:1963}
B.D. Coleman and W.~Noll.
\newblock The thermodynamics of elastic materials with heat conduction and
  viscosity.
\newblock {\em Arch. Rational Mech. Anal.}, 13:167--178, 1963.

\bibitem{Coleman:1967}
B.~D. Coleman and M.~E. Gurtin.
\newblock Thermodynamics with internal state variables.
\newblock {\em J. Chem. Phys.}, 47(2):597--613, 07 1967.

\bibitem{lubliner1972a}
J.~Lubliner.
\newblock On the thermodynamic foundations of non-linear solid mechanics.
\newblock {\em Int. J. Non-Linear Mech}, 7:237--254,, 1972.

\bibitem{Graffi:1928}
D.~Graffi.
\newblock Sui problema della eredit{\'a} lineare.
\newblock {\em Nuovo Cim.}, 5:53--71, 1928.

\bibitem{Gurtin:1965}
M.E. Gurtin and I.~Herrera.
\newblock On dissipation inequalities and linear viscoelasticity.
\newblock {\em Q. J. Mech. Appl. Math.}, 23:235--245, 1965.

\bibitem{Fabrizio:1988}
M.~Fabrizio and A.~Morro.
\newblock Viscoelastic relaxation functions compatible with thermodynamics.
\newblock {\em J. Elast.}, 19:63--75, 1988.

\bibitem{Matarazzo:2001}
G.~Matarazzo.
\newblock Irreversibility of time and symmetry property of relaxation function
  in linear viscoelasticity.
\newblock {\em Mech. Res. Commun.}, 28(4):373--380, 2001.

\bibitem{Rogers:1963}
G.~Rogers and A.C. Pipkin.
\newblock Asymmetric relaxation and compliance matrices in linear
  viscoelasticity.
\newblock {\em Z. Angew. Math. Phys}, 14:334--343, 1963.

\bibitem{Biot:1954}
M.A. Biot.
\newblock Theory of stress-strain relations in anisotropic viscoelasticity and
  relaxation phenomena.
\newblock {\em J. Appl. Phys}, 25:1385--1391, 1954.

\bibitem{Biot:1958}
M.A. Biot.
\newblock Linear thermodynamics and the mechanics of solids.

\bibitem{Staverman:1952}
A.J. Staverman and F.~Schwarzl.
\newblock Non-equilibrium thermodynamics of viscoelastic behavior.
\newblock {\em Proc. Roy. Ac. Sci. Amsterdam}, 55:474–490, 1952.

\bibitem{Staverman:1954}
A.J. Staverman.
\newblock Thermodynamics o[ linear viscoelastic behavior.
\newblock In {\em Proceedings of the Second International Congress on Rheology
  (1953, Oxford, UK)}, pages 134--138, New York, 1954. British Society of
  Rheology, Academic Press.

\bibitem{Meixner:1953}
J.~Meixner.
\newblock Die thermodynamische theorie der relaxationserscheinungen und ihr
  zusammenhang mit der nachwirkungstheorie.
\newblock {\em Koll. Zeitsch.}, 134:3--20, 1953.

\bibitem{Meixner:1954}
J.~Meixner.
\newblock Thermodynamische theorie der elastischen relaxation.
\newblock {\em Z. Naturforsch. A.}, 718:654--665, 1954.

\bibitem{Groot:1952}
S.R. {de Groot}.
\newblock {\em Thermodynamics of Irreversible Processes}.
\newblock Interscience Publishers Inc, New York, 1952.

\bibitem{Coleman:1960}
B.D. Coleman and C.~Truesdell.
\newblock {\em J. Chem. Phys.}, 33(1):28--31, 1960.

\bibitem{Eringen:1960}
A.~C. Eringen.
\newblock Irreversible thermodynamics and continuum mechanics.
\newblock {\em Phys. Rev.}, 117:1174--1183, Mar 1960.

\bibitem{Dafermos:1970}
C.~M. Dafermos.
\newblock Asymptotic stability in viscoelasticity.
\newblock {\em Arch. Ration. Mech. Anal.}, 37:297--308, 1970.

\bibitem{fichera1979b}
G.~Fichera.
\newblock Avere una memoria tenace crea gravi problemi.
\newblock {\em Arch. Ration. Mech. Anal.}, 70(2):101--112, 1979.

\bibitem{gripenberg:1990}
G.~Gripenberg, S.O. Londen, and O.~Staffans.
\newblock {\em Volterra Integral and Functional Equations}.
\newblock Encyclopedia of Mathematics and its Applications. Cambridge
  University Press, Cambridge.

\bibitem{Mizel:1966}
V.J. Mizel and C.C. Wang.
\newblock A fading memory hypothesis which suffices for chain rules.
\newblock {\em Arch. Ration. Mech. Anal.}, 23:124--134, 1966.

\bibitem{Conway:1990}
J.~B. Conway.
\newblock {\em A Course in Functional Analysis. Graduate Texts in Mathematics},
  volume~96 of {\em Graduate Texts in Mathematics}.
\newblock Springer-Verlag, New York, 1990.

\bibitem{Mendelson:1968}
A.~Mendelson.
\newblock {\em Plasticity: Theory and application}.
\newblock Macmillan series in applied mathematics. Mcmillan, New York, 1968.

\bibitem{Temam:1979}
Roger Temam.
\newblock {\em Navier--{S}tokes equations}, volume~2 of {\em Studies in
  Mathematics and its Applications}.
\newblock North-Holland Publishing Co., Amsterdam-New York, revised edition,
  1979.

\bibitem{Conti:2018}
S.~Conti, S.~M\"uller, and M.~Ortiz.
\newblock Data-driven problems in elasticity.
\newblock {\em Arch. Ration. Mech. Anal.}, 229:79--123, 2018.

\bibitem{Ciarlet:2010}
P.G. Ciarlet.
\newblock On {K}orn’s inequality.
\newblock {\em Chin. Ann. Math. Series B}, 31:607--618, 2010.

\bibitem{Tallec:1990}
P.~L. Tallec.
\newblock {\em Numerical Analysis of Viscoelastic Problems}.
\newblock Recherches en math{\'e}matiques appliqu{\'e}es. Masson, Paris, 1990.

\bibitem{Reese:1998}
S.~Reese and S.~Govindjee.
\newblock A theory of finite viscoelasticity and numerical aspects.
\newblock {\em Int. J. Solids Struct.}, 35(26):3455--3482, 1998.

\bibitem{Govindjee:2004}
S.~Govindjee.
\newblock {\em Numerical Issues in Finite Elasticity and Viscoelasticity},
  pages 187--232.
\newblock Springer Vienna, Vienna, 2004.

\bibitem{Pinkus:1989}
A.~Pinkus.
\newblock {\em On L1 approximation}, volume~93 of {\em Cambridge tracts in
  mathematics}.
\newblock Cambridge University Press, Cambridge, 1989.

\bibitem{Kestin:1970}
J.~Kestin and J.~R. Rice.
\newblock {\em Paradoxes in the Application of Thermodynamics to Strained
  Solids}, pages 275--298.
\newblock A Critica1 Review of Thermo-dynamics. Mono Book Corp., Baltimore,
  1970.

\bibitem{Bogachev:2007}
V.~I. Bogachev.
\newblock {\em Measure Theory}, volume~I.
\newblock Springer-Verlag, Berlin-Heidelberg, 2007.

\bibitem{Horstemeyer:2010}
Mark~F. Horstemeyer and Douglas~J. Bammann.
\newblock Historical review of internal state variable theory for inelasticity.
\newblock {\em Int. J. Plast.}, 26(9):1310--1334, 2010.

\bibitem{Brezis:2010}
H.~Brezis.
\newblock {\em Functional Analysis, Sobolev Spaces and Partial Differential
  Equations}.
\newblock Universitext. Springer, New York.

\bibitem{Eckart:1940}
C.~Eckart.
\newblock Thermodynamics of irreversible processes, {I}. {T}he simple fluid.
\newblock {\em Phys. Rev. A}, 58(4):267--269, 1940.

\bibitem{Eckart:1948}
C.~Eckart.
\newblock The thermodynamics of irreversible processes. iv. the theory of
  elasticity and anelasticity.
\newblock {\em Physical Review}, 73(4):373--382, 1948.

\bibitem{Ziegler:1958}
H.~Ziegler.
\newblock An attempt to generalise {O}nsager's principle, and its significance
  for rheological problems.
\newblock {\em Z. Angew. Math. Phys.}, 9(6):748--763, 1958.

\bibitem{Rice:1971}
J.~R. Rice.
\newblock Inelastic constitutive relations for solids: an internal-variable
  theory and its application to metal plasticity.
\newblock {\em J. Mech. Phys. Solids}, 19(6):433--455, 1971.

\bibitem{lubliner1973a}
J.~Lubliner.
\newblock On the structure of the rate equations of materials with internal
  variables.
\newblock {\em Acta Mech.}, 17:109--119, 1973.

\bibitem{Rice:1975}
J.~R. Rice.
\newblock {\em Continuum Mechanics and Thermodynamics of Plasticity in Relation
  to Microscale Deformation Mechanisms}, chapter 2 of Constitutive Equations in
  Plasticity, Argon, A.~S. (ed.), pages 23--79.
\newblock MIT Press, 1975.

\bibitem{Sion:1958}
M.~Sion.
\newblock On general minimax theorems.
\newblock {\em Pacific J. Math.}, 8(1):171--176, 1958.

\end{thebibliography}

\end{document}